\documentclass[
  aps,
  prx,
  reprint,
  superscriptaddress,
  nofootinbib,
  longbibliography,
  floatfix
]{revtex4-2}
\usepackage{braket}
\usepackage[T1]{fontenc}
\usepackage[utf8]{inputenc}
\usepackage{lmodern}
\usepackage{amsmath,amssymb,mathtools}
\usepackage{bm}
\usepackage{graphicx}
\usepackage{microtype}
\usepackage{amsthm}
\usepackage[colorlinks=true,allcolors=blue]{hyperref}
\usepackage{dsfont}
\usepackage{xcolor}

\newcommand{\trace}{\operatorname{Tr}}
\newcommand{\ketbra}[2]{\ket{#1}\!\bra{#2}}
\newcommand{\Ic}{I_{\mathrm c}}

\DeclareRobustCommand{\rchi}{{\mathpalette\irchi\relax}}
\newcommand{\irchi}[2]{\raisebox{\depth}{$#1\chi$}}

\theoremstyle{plain}
\newtheorem{theorem}{Theorem}
\newtheorem{proposition}[theorem]{Proposition}
\newtheorem{lemma}[theorem]{Lemma}
\newtheorem{corollary}[theorem]{Corollary}
\theoremstyle{definition}
\newtheorem{remark}[theorem]{Remark}

\begin{document}

\title{Sharp Quantum Capacity Thresholds: Exponential Strong Converses for Degradable and Antidegradable Channels}

\author{Tulja Varun Kondra}
\email{tuljavarun@gmail.com}
\author{Raphael Brinster}
\author{Hermann Kampermann}
\author{Dagmar Bruß}
\author{Nikolai Wyderka}
\affiliation{Institut für Theoretische Physik III,
Heinrich-Heine-Universität Düsseldorf, Universitätsstraße 1,
40225 Düsseldorf, Germany}
\date{\today}

\begin{abstract}
The quantum capacity of a noisy channel quantifies the maximum rate at which quantum information can be transmitted reliably. For general channels, its evaluation requires an optimization over arbitrarily many channel uses. Degradable channels form a central exception: their capacity is given by the single-letter coherent information, while antidegradable channels have zero capacity. Nevertheless, even for these fundamental classes, it has remained open whether communication above capacity becomes possible when a fixed non-maximal error is tolerated. Here we resolve this problem by proving an exponential strong converse for every finite-dimensional degradable and antidegradable channel: at any rate above capacity, the fidelity of every coding scheme decays exponentially with the number of channel uses. As an immediate consequence, we establish the first all-code exponential strong converse for the quantum erasure channel throughout its full parameter range, strengthening previous results that applied only to almost all codes. We also show that exponential strong-converse bounds are preserved under receiver post-processing. This yields efficiently computable semidefinite-programming bounds for arbitrary finite-dimensional channels, improved bounds for Pauli channels, and an exact exponential strong converse for a nondegradable multilevel amplitude-damping family.

\end{abstract}

\maketitle

% =====================================================================
\section{Introduction}
\label{sec:introduction}
% =====================================================================

Reliable quantum communication is a basic resource for distributing
entanglement, implementing nonlocal gates, and building quantum
networks~\cite{Bennett1993Teleportation,Cirac1999Distributed,Wehner2018QuantumInternet}.
In analogy with Shannon's
theory~\cite{Shannon1948}, the basic question is how much quantum
information survives per use of a noisy channel $\mathcal N$. An $(n,M)$
code transmits an $M$-dimensional quantum system through $n$ independent
uses of $\mathcal N$ and therefore operates at rate
\begin{equation}
  r=\frac{1}{n}\log M
  \label{eq:intro-rate}
\end{equation}
qubits per channel use, and the quantum capacity $Q(\mathcal N)$ is the
supremum of the rates at which the transmission error can be made to vanish
as $n\to\infty$. The Lloyd-Shor-Devetak theorem identifies
$Q(\mathcal N)$ with the regularized coherent
information~\cite{Lloyd1997,Devetak2005}. This characterization is
operationally complete but requires an optimization over unboundedly many
channel uses and has no known finite-letter characterization in general
\cite{KretschmannWerner2004}.

Two structural classes escape the regularization, and they do so for
complementary physical reasons. A channel is \emph{degradable} when the
receiver can reconstruct, from their own output alone, everything that
leaked into the environment. Nothing is available to the environment that
the receiver does not already hold, the coherent information becomes
additive, and the capacity collapses to a single-letter
formula~\cite{DevetakShor2005}. A channel is \emph{antidegradable} when the
opposite holds: the environment can reconstruct the receiver's output. The
channel can then realize two systems with identical receiver-channel marginals,
a configuration incompatible with reliable transmission of unknown quantum
information; consequently its capacity is zero~\cite{WolfPerezGarcia2007}. The quantum
erasure channel interpolates between the two,
being degradable below and antidegradable above erasure probability
$1/2$~\cite{BennettDiVincenzoSmolin1997}. Figure~\ref{fig:degradable-antidegradable-extension}(a),(b)
summarizes the two definitions.

Knowing $Q(\mathcal N)$ does not by itself settle the operational status of
the threshold, because the capacity is defined by requiring the error to
vanish. Nothing in that definition rules out a scheme that transmits at a
rate above $Q(\mathcal N)$ while incurring a fixed error of, say, $10\%$
forever. A \emph{strong converse} excludes this: it asserts that at every
asymptotic rate above the capacity the error tends to its maximal value, or
equivalently that the entanglement fidelity $F_n$ tends to zero. An
\emph{exponential} strong converse sharpens the conclusion to
$F_n\le 2^{-\gamma n}$ for some $\gamma>0$, making the capacity a sharp
threshold with a quantitative penalty for exceeding it.

For unassisted quantum communication, a full strong converse has remained elusive. Morgan and Winter proved a ``pretty strong'' converse for degradable channels: above capacity, the asymptotic error is not zero, but need not go to one~\cite{MorganWinter2014}. They further reduced the full strong-converse conjecture for degradable channels to establishing a strong converse for symmetric zero-capacity channels, which remained open. Tomamichel, Wilde, and Winter showed that the Rains information provides an exponential strong-converse upper bound on the achievable rate for every channel~\cite{TomamichelWildeWinter2017}. This yields a full strong converse whenever the Rains information coincides with the quantum capacity, as for generalized dephasing channels, but generally leaves a gap. For the erasure channel, strong converses were known only for maximally entangled inputs~\cite{SharmaWarsi2013} or almost all codes~\cite{WildeWinter2014}, rather than arbitrary codes. Antidegradable channels were known to satisfy a strong converse for private communication, but only a pretty strong converse for quantum communication~\cite{KaurEtAl2021,KhanianHirche2025}. Most recently, an exponential strong converse for Pauli channels was obtained within the restricted class of stabilizer codes~\cite{Tomamichel2026}.

\emph{Results.}---We prove an all-code exponential strong converse for every
finite-dimensional antidegradable channel (Theorem~\ref{thm:antidegradable}),
in the generalized form required by the Morgan-Winter reduction. Feeding
this into that reduction yields an exponential strong converse for every
finite-dimensional degradable channel at its single-letter capacity
(Theorem~\ref{thm:degradable}), resolving the conjecture of
Ref.~\cite{MorganWinter2014}. Because the erasure channel is degradable or
antidegradable at every erasure probability, the two results cover it
completely and upgrade the almost-all-codes statement of
Ref.~\cite{WildeWinter2014} to an all-code one
(Corollary~\ref{cor:erasure}). We then show that exponential
strong-converse rates are inherited under receiver post-processing
(Proposition~\ref{prop:postprocessing}), a closure property that converts
any degradable extension of a channel into an all-code exponential bound.
Applying this construction to the largest antidegradable component of an
arbitrary channel gives an efficiently computable semidefinite-programming
bound (Theorem~\ref{thm:universal}). Furthermore, we derive an explicit Pauli-channel specialization. Receiver
post-processing also establishes an all-code exponential converse at the
exact capacity of a nondegradable multilevel amplitude-damping family.

\emph{Method.}---The technical core is a quantitative implementation of
no-cloning. An antidegradable channel admits a symmetric two-output
extension, so each of the $n$ channel uses can be made to produce two
outputs with the same marginal statistics as the receiver's. Every binary
string $x\in\{0,1\}^n$ then selects one output per use and defines a
placement of the \emph{same} decoder on a different set of systems; see
Fig.~\ref{fig:degradable-antidegradable-extension}(c). All $2^n$ placements
succeed with exactly the same fidelity, so the fidelity is bounded by the
operator norm of their uniform average. On the pairs selected by our signed averaging kernel, two placements share
few outputs and therefore have small overlap, controlled by the dimension of
the shared systems relative to the code dimension. What remains is a
Boolean-analysis problem: to certify that a family of projectors with
controlled \emph{pairwise} overlaps has a small average, we build a weight
matrix supported only on pairs that share at most $t$
coordinates, from a low-degree polynomial approximation to the binary $\mathrm{NOR}_n$ function.
The approximate-degree bound then converts local incompatibility into a
global, exponentially small bound. The argument places no restriction on
the encoder or on the decoder, which may be arbitrary and collective.

% =====================================================================
\section{Setup}
\label{sec:setup}
% =====================================================================

We work with finite-dimensional systems throughout. For a system $A$ we
write $\mathcal L(A)$ for the linear operators on $A$ and $|A|:=\dim A$,
and all logarithms are base two unless written as $\ln$. Following the
standard convention, $A'$ denotes the input system of a channel and a channel $\mathcal N:A'\to B$ is a
completely positive trace-preserving map from $\mathcal L(A')$ to
$\mathcal L(B)$.

\subsection{Complementary channels, degradability, symmetric extensions}

Every channel has a Stinespring representation~\cite{Stinespring1955}: there
is an environment $E$ and an isometry $V:A'\to BE$ with
\begin{equation}
  \mathcal N(\rho)=\trace_E\!\left[V\rho V^\dagger\right].
  \label{eq:stinespring}
\end{equation}
The complementary channel collects what the environment learns,
\begin{equation}
  \mathcal N^{c}(\rho)=\trace_B\!\left[V\rho V^\dagger\right].
  \label{eq:complementary}
\end{equation}
Different Stinespring representations change $\mathcal N^c$ only by an
isometry on $E$, which affects none of the definitions below.

$\mathcal N$ is \emph{degradable} if there is a channel
$\mathcal D:B\to E$ with
\begin{equation}
  \mathcal N^{c}=\mathcal D\circ\mathcal N ,
  \label{eq:degradable}
\end{equation}
and \emph{antidegradable} if there is a channel $\mathcal A:E\to B$ with
\begin{equation}
  \mathcal N=\mathcal A\circ\mathcal N^{c}.
  \label{eq:antidegradable}
\end{equation}
We call a channel \emph{self-complementary} if it admits a Stinespring
isometry $W:A'\to D_0D_1$, with $D_0\simeq D_1$, whose marginals coincide
under a fixed identification. Such a channel is both degradable and
antidegradable. This is the stronger notion called ``symmetric'' in the
Morgan-Winter construction used below.

Antidegradability has an equivalent form that we use throughout. A channel
$\mathcal N:A'\to B$ is antidegradable if and only if it admits a two-output
extension $\widetilde{\mathcal N}:A'\to B_0B_1$, with $B_0\simeq B_1\simeq B$,
whose two marginals both reproduce it,
\begin{equation}
  \trace_{B_1}\circ\,\widetilde{\mathcal N}
  =\trace_{B_0}\circ\,\widetilde{\mathcal N}
  =\mathcal N ,
  \label{eq:two-output-extension}
\end{equation}
and the extension may always be taken invariant under exchange of $B_0$ and
$B_1$~\cite{MyhrLutkenhaus2009,WolfPerezGarcia2007}. One way to see the
channel statement is to apply the symmetric-extension criterion to the
normalized Choi state and use the Choi-Jamio\l kowski correspondence.
Equation~\eqref{eq:two-output-extension} is an equality of marginal channels,
including on inputs entangled with a reference; it does not assert that the
two outputs are independent clones. 

\begin{figure*}[t]
    \centering
    \includegraphics[width=\textwidth]{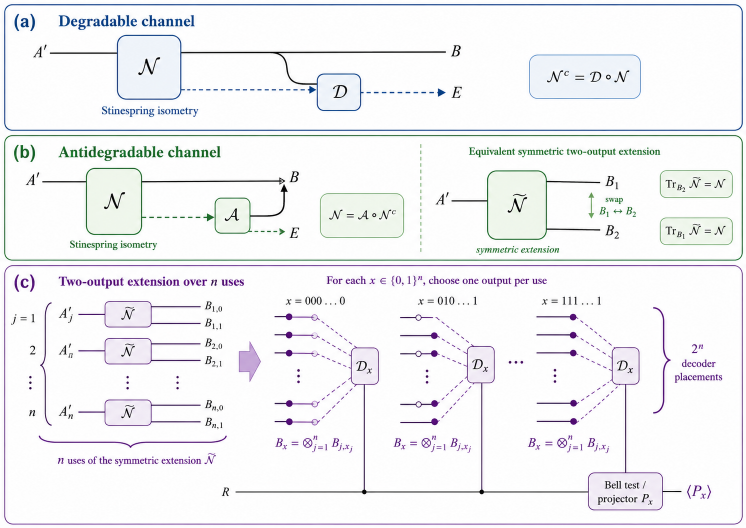}
    \caption{\textbf{Degradable and antidegradable channels, and the
    decoder-placement mechanism.}
    \textbf{(a)} $\mathcal N$ is degradable when the environment output can
    be simulated from the receiver output by a degrading channel
    $\mathcal D$, so that $\mathcal N^{c}=\mathcal D\circ\mathcal N$.
    \textbf{(b)} $\mathcal N$ is antidegradable when the receiver output can
    be simulated from the environment, $\mathcal N=\mathcal A\circ\mathcal N^{c}$;
    equivalently, it admits a symmetric two-output extension
    $\widetilde{\mathcal N}:A'\to B_0B_1$ whose two marginal channels both
    equal $\mathcal N$.
    \textbf{(c)} Applying the symmetric extension independently to $n$
    channel uses produces two receiver-like outputs at each position. Every
    string $x\in\{0,1\}^{n}$ selects one output from each pair and hence
    defines a subsystem $B_{x}=\bigotimes_{j=1}^{n}B_{j,x_{j}}$ on which a
    copy $\mathcal D_{x}$ of the original decoder may be placed. This yields
    $2^{n}$ overlapping placements and corresponding fidelity projectors
    $P_{x}$, all with the same success probability, which form the basis of
    the argument.}
    \label{fig:degradable-antidegradable-extension}
\end{figure*}

\subsection{Codes and figures of merit}

We phrase the converse in terms of entanglement generation, the task in
which the sender establishes a maximally entangled state with the receiver.
An entanglement-transmission code gives such a scheme by choosing a maximally
entangled input, while the standard expurgation and fidelity estimates relate
subspace transmission to entanglement transmission without changing the
asymptotic capacity~\cite{BarnumKnillNielsen2000,Devetak2005}. Entanglement
generation is also the formulation used by the Morgan-Winter reduction.

Let $R$, $S$, $\widehat S$ be $M$-dimensional systems and write
\begin{equation}
  \Phi_M:=\ketbra{\Phi_M}{\Phi_M},
  \qquad
  \ket{\Phi_M}:=\frac{1}{\sqrt M}\sum_{i=1}^{M}\ket{i}_R\ket{i}_S .
  \label{eq:mes}
\end{equation}
An $(n,M)$ entanglement-generation code for $\mathcal N:A'\to B$ consists of
an encoder $\mathcal E_n:S\to A'^{\,n}$ and a decoder
$\mathcal D_n:B^n\to\widehat S$. The $n$ channel uses act in parallel on the
encoded state, and the final state is
\begin{equation}
  \omega^{(n)}_{R\widehat S}
  :=\left(\operatorname{id}_R\otimes\,
  \mathcal D_n\circ\mathcal N^{\otimes n}\circ\mathcal E_n\right)(\Phi_M).
  \label{eq:decoded-state}
\end{equation}
The target is a maximally entangled state of the same dimension, which we
write $\Phi_{R\widehat S}$ to record that it now lives on $R\widehat S$
rather than on $RS$. It is often convenient to absorb the encoder into the
input state,
\begin{equation}
  \rho_{RA'^{\,n}}:=\left(\operatorname{id}_R\otimes\mathcal E_n\right)(\Phi_M),
  \qquad \rho_R=\frac{I_R}{M}.
  \label{eq:code-state}
\end{equation}

We also need the slightly broader task introduced by Morgan and
Winter~\cite{MorganWinter2014}. A \emph{generalized $(n,M)$ Bell-target
scheme} is a pair $(\rho_{RA'^{\,n}},\mathcal D_n)$ consisting of an
arbitrary state on $RA'^{\,n}$ with $|R|=M$, together with a decoder
$\mathcal D_n:B^n\to\widehat S$; the figure of merit is still the overlap of
the output with $\Phi_{R\widehat S}$, but the marginal $\rho_R$ need no
longer be maximally mixed. Standard codes are the special case
\eqref{eq:code-state}. The generalization is not cosmetic: the reduction of
Ref.~\cite{MorganWinter2014} converts a code for a degradable channel into a
scheme of exactly this kind for an associated self-complementary channel, so a
converse for degradable channels can only be extracted from a converse that
holds in the generalized setting. Our argument costs nothing extra here,
since it never uses the form of $\rho_{RA'^{\,n}}$. Morgan and Winter state
the corresponding optimization for pure test states; the two versions agree
because the Bell overlap is linear in $\rho_{RA'^{\,n}}$, so at least one pure
state in any convex decomposition performs at least as well as the mixture.

With $F(\rho,\sigma):=\lVert\sqrt\rho\sqrt\sigma\rVert_1^2$ and purified
distance $P(\rho,\sigma):=\sqrt{1-F(\rho,\sigma)}$, and because the target
is pure, the entanglement fidelity of a scheme is
\begin{equation}
  F_n:=F\!\left(\omega^{(n)}_{R\widehat S},\Phi_{R\widehat S}\right)
  =\trace\!\left[\Phi_{R\widehat S}\,\omega^{(n)}_{R\widehat S}\right],
  \label{eq:fidelity}
\end{equation}
and its error is $\varepsilon_n:=\sqrt{1-F_n}$. Thus $\varepsilon_n\to0$
means $F_n\to1$ and $\varepsilon_n\to1$ means $F_n\to0$. Note that $F$ is
already the squared fidelity, so no separate symbol is needed. The rate of
an $(n,M_n)$ scheme is
\begin{equation}
  r_n:=\frac{1}{n}\log M_n .
  \label{eq:code-rate}
\end{equation}
Finally, $N_{\mathrm E}(n,\varepsilon\mid\mathcal N)$ denotes the largest
code dimension $M$ for which some generalized $(n,M)$ Bell-target scheme
achieves purified-distance error at most $\varepsilon$. This is the same
error convention used in the Morgan-Winter bound below.

\subsection{Fixed-error and strong-converse capacities}

For $0\le\varepsilon<1$ the fixed-error quantum capacity is \cite{MorganWinter2014}
\begin{equation}
  Q^\varepsilon(\mathcal N)
  :=\sup\left\{\liminf_{n\to\infty}r_n:\
  \limsup_{n\to\infty}\varepsilon_n\le\varepsilon\right\},
  \label{eq:fixed-error-capacity}
\end{equation}
the supremum running over sequences of standard entanglement-generation
codes; the usual capacity is $Q(\mathcal N)=Q^0(\mathcal N)$, and the
strong-converse capacity is
$Q^\dagger(\mathcal N):=\sup_{0\le\varepsilon<1}Q^\varepsilon(\mathcal N)$.
Above $Q^\dagger(\mathcal N)$ the error is driven to one no matter how large
a fixed non-maximal error is tolerated, and the channel has the
strong-converse property when $Q^\dagger(\mathcal N)=Q(\mathcal N)$.

We control the speed as well. A number $r$ is an \emph{exponential
strong-converse rate} for $\mathcal N$ if for every $r'>r$ there is a
$\gamma>0$ such that every sequence of codes with
$\liminf_n r_n\ge r'$ obeys
\begin{equation}
  F_n\le 2^{-\gamma n}
  \label{eq:exponential-sc}
\end{equation}
for all large enough $n$, and $Q^{\exp(\dagger)}(\mathcal N)$ is the
infimum of such $r$. By construction
\begin{equation}
  Q(\mathcal N)\le Q^\dagger(\mathcal N)\le Q^{\exp(\dagger)}(\mathcal N),
  \label{eq:capacity-hierarchy}
\end{equation}
and an exponential strong converse at the capacity forces equality
throughout.

\subsection{Coherent information}

Let $V:\mathcal H_{A'}\to\mathcal H_B\otimes\mathcal H_E$ be a Stinespring
isometry of $\mathcal N$, let $\psi_{RA'}$ purify an input state
$\rho_{A'}$, and let
$\omega_{RBE}:=(\operatorname{id}_R\otimes V)\psi_{RA'}
(\operatorname{id}_R\otimes V^\dagger)$. With $H(X)_\omega$ the von Neumann
entropy of the reduced state on $X$, the coherent information is
\begin{equation}
  \Ic(\rho_{A'},\mathcal N):=H(B)_\omega-H(E)_\omega=-H(R|B)_\omega ,
  \label{eq:coherent-information}
\end{equation}
the second equality holding because $\omega_{RBE}$ is pure, and
\begin{equation}
  Q^{(1)}(\mathcal N):=\max_{\rho_{A'}}\Ic(\rho_{A'},\mathcal N).
  \label{eq:one-letter}
\end{equation}
The Lloyd-Shor-Devetak theorem gives
$Q(\mathcal N)=\lim_{k}\frac1k Q^{(1)}(\mathcal N^{\otimes k})$
\cite{Lloyd1997,Devetak2005}, which collapses to
$Q(\mathcal N)=Q^{(1)}(\mathcal N)$ for degradable
channels~\cite{DevetakShor2005} and to $Q(\mathcal N)=0$ for antidegradable
ones. Neither formula says anything about $Q^\dagger$ or
$Q^{\exp(\dagger)}$, which is the gap this paper closes.

% =====================================================================
\section{Antidegradable channels}
\label{sec:antidegradable}
% =====================================================================

\begin{theorem}\label{thm:antidegradable}
Let $\mathcal N:A'\to B$ be antidegradable and finite-dimensional. For every
$r>0$ there are constants $\gamma_r>0$ and $n_r$ such that every
generalized $(n,M_n)$ Bell-target scheme with $M_n\ge 2^{nr}$ obeys
\begin{equation}
  F_n\le 2^{-\gamma_r n},\qquad n\ge n_r .
  \label{eq:antidegradable-bound}
\end{equation}
In particular $Q^{\exp(\dagger)}(\mathcal N)=0$, for arbitrary encoders and
arbitrary collective decoders.
\end{theorem}

The theorem follows from a finite-block estimate, which we state separately
because the degradable case in Sec.~\ref{sec:degradable} needs it in
quantitative form.

\begin{proposition}\label{prop:finite-block}
There are universal constants $a,C_{\mathrm{app}}>0$ such that the following holds. Let
$\mathcal N:A'\to B$ be antidegradable, put $d:=|B|$, and let $t$ be an
integer with
\begin{equation}
  4C_{\mathrm{app}}\sqrt n\le t\le \frac n2 .
  \label{eq:t-range}
\end{equation}
Then every generalized $(n,M)$ Bell-target scheme for $\mathcal N$ satisfies
\begin{equation}
  F_n\le
  2\exp\!\left(-\frac{a t^2}{n}\right)
  +\left(\sum_{j=0}^{t}\binom nj\right)^{1/4}
  \min\left\{1,\frac{d^{\,t}}{M}\right\}^{1/2}.
  \label{eq:finite-block}
\end{equation}
\end{proposition}

We sketch the idea of the proof here and defer the details to the Supplemental
Material.

\emph{Exponentially many decoder placements.}---Let
$\widetilde{\mathcal N}:A'\to B_0B_1$ be a symmetric two-output extension
and apply it to each of the $n$ uses, producing the state
\begin{align}
  \omega_{R\mathbf B}
  &:=\left(\operatorname{id}_R\otimes\widetilde{\mathcal N}^{\otimes n}\right)
  (\rho_{RA'^{\,n}}),\nonumber\\
  \mathbf B&:=\bigotimes_{j=1}^{n}\left(B_{j,0}\otimes B_{j,1}\right).
  \label{eq:extended-output-state}
\end{align}
Each string $x\in\{0,1\}^n$ selects
$B_x:=\bigotimes_{j=1}^{n}B_{j,x_j}$, and by
Eq.~\eqref{eq:two-output-extension}
\begin{equation}
  \trace_{\mathbf B\setminus B_x}\left[\omega_{R\mathbf B}\right]
  =\left(\operatorname{id}_R\otimes\mathcal N^{\otimes n}\right)
  (\rho_{RA'^{\,n}})
  \label{eq:marginal-invariance}
\end{equation}
for \emph{every} $x$: each selection reproduces exactly the output the
receiver would have seen. A copy $\mathcal D_x$ of the original decoder may
therefore be placed on $B_x$, for all $2^n$ strings at once.

\emph{Fidelity projectors.}---Give each placement a private ancilla $G_x$
prepared in $\ket 0$ and a unitary dilation
$U_x:B_xG_x\to\widehat S E_x$ of $\mathcal D_x$. Testing the decoded state
against $\Phi_{R\widehat S}$ then amounts to measuring the projector
\begin{equation}
  P_x:=(I_R\otimes U_x^\dagger)
  \left(\Phi_{R\widehat S}\otimes I_{E_x}\right)
  (I_R\otimes U_x),
  \label{eq:fidelity-projector}
\end{equation}
extended by the identity on the outputs and ancillas that the $x$-th
placement does not touch, so that all $P_x$ act on the common space
\begin{equation}
  \mathcal H_{\mathrm{tot}}
  :=\mathcal H_R\otimes
  \bigotimes_{j=1}^{n}\left(\mathcal H_{B_{j,0}}\otimes\mathcal H_{B_{j,1}}\right)
  \otimes\bigotimes_{y\in\{0,1\}^{n}}\mathcal H_{G_y}.
\end{equation}
Writing $\widetilde\omega$ for the state
$\omega_{R\mathbf B}\otimes\bigotimes_y\ketbra{0}{0}_{G_y}$,
Eq.~\eqref{eq:marginal-invariance} gives
$F_n=\trace[P_x\widetilde\omega]$ for every $x$, whence for the uniform
average $A:=2^{-n}\sum_x P_x$,
\begin{equation}
  F_n=\trace\left[A\,\widetilde\omega\right]\le\lVert A\rVert_\infty .
  \label{eq:fidelity-operator-norm}
\end{equation}
Every placement succeeds with the same probability, so the fidelity is
squeezed by how incompatible the placements are with one another.

\emph{Pairwise incompatibility.}---Two placements interact only through the
outputs they share. Let
$s(x,y):=|\{j:x_j=y_j\}|$. In Lemma~\ref{lem:sm-overlap} in the Supplemental Material we show, by
expanding both maximally entangled states and using unitarity of $U_x,U_y$
block by block on the shared systems,
for every pair of distinct strings $x\ne y$,
\begin{equation}
  \lVert P_xP_y\rVert_\infty
  \le\min\left\{1,\frac{d^{\,s(x,y)}}{M}\right\}.
  \label{eq:projector-overlap}
\end{equation}
This is the quantitative no-cloning statement: two decoders that overlap on
few channel uses cannot both recover a large entangled system, and the
obstruction is governed by the shared dimension $d^{\,s(x,y)}$ in relation to the
code dimension $M$.

\emph{From pairwise to global.}---Equation~\eqref{eq:projector-overlap}
constrains pairs, whereas Eq.~\eqref{eq:fidelity-operator-norm} asks about
all $2^n$ projectors simultaneously; nothing forbids a large average built
from projectors that are only pairwise nearly orthogonal. We bridge the gap
with a weight matrix indexed by bit strings in $\{0,1\}^n$: Suppose a real symmetric 
$W\in\mathbb R^{2^n\times2^n}$ satisfies $W\mathbf 1=\mathbf 1$,
$\lVert W-U\rVert_\infty\le\zeta$ with
$U_{xy} = 2^{-n}$,
$\max_x\sum_y|W_{xy}|\le L$, and $W_{xy}\neq0$ only for pairs with
$\lVert P_xP_y\rVert_\infty\le\eta$. From such $W$, we derive in Lemma~\ref{lem:sm-average} the bound
\begin{equation}
  \lVert A\rVert_\infty\le\zeta+\sqrt{L\eta}.
  \label{eq:W-bound}
\end{equation}
We show in Lemma~\ref{lem:sm-weight} that such a $W$ exists whenever we are willing to discard pairs sharing more than
$t$ coordinates. To that end, we make use of the construction in Ref.~\cite{deWolfNote} of a multilinear polynomial $q$ of degree at most
$t$ with $q(0)=1$ and $|q(z)|\le2\varepsilon$ for $z\neq0$, which exists for
$\varepsilon=\exp(-t^2/4C_{\mathrm{app}}^2n)$, and set
$W_{xy}:=\widehat q(\mathds 1\oplus x\oplus y)$ with $\widehat q$ the
normalized Hadamard transform. Bounded degree makes $\widehat q$ supported
on strings of weight at most $t$, which is exactly the constraint that
$W_{xy}\neq0$ only when $x$ and $y$ agree in at most $t$ places. In
particular $W_{xx}=0$, because $s(x,x)=n>t$, so the diagonal case excluded
from Eq.~\eqref{eq:projector-overlap} is never used. Finally, we show that $W$ and $U$ can be diagonalized simultaneously, giving
$\zeta\le2\varepsilon$; and Cauchy-Schwarz with Parseval gives
$L\le(\sum_{j\le t}\binom nj)^{1/2}$. Substituting into
Eq.~\eqref{eq:W-bound} with $\eta=\min\{1,d^{\,t}/M\}$ proves
Proposition~\ref{prop:finite-block}. 

With Proposition~\ref{prop:finite-block} in place, we can prove Theorem~\ref{thm:antidegradable}:
\begin{proof}[Proof of Theorem~\ref{thm:antidegradable}]
Fix $r>0$ and choose $0<\delta<1/2$ such that
\begin{equation}
  \delta\log d+\tfrac12 h_2(\delta)<r ,
  \label{eq:delta-choice}
\end{equation}
where $h_2$ is the binary entropy. Set $t=\lfloor\delta n\rfloor$. For all
sufficiently large $n$, Eq.~\eqref{eq:t-range} holds and
$t\ge\delta n/2$. Since $t/n\le\delta<1/2$, monotonicity of $h_2$ on
$[0,1/2]$ gives
$\sum_{j\le t}\binom nj\le2^{nh_2(t/n)}\le2^{nh_2(\delta)}$, while
$d^t\le d^{\delta n}$. Proposition~\ref{prop:finite-block} and
$M_n\ge2^{nr}$ therefore imply
\begin{equation}
  F_n\le
  2\cdot 2^{-a\delta^2 n\log e/4}
  +2^{-\frac n2\left[r-\delta\log d-\frac12 h_2(\delta)\right]}.
\end{equation}
Both exponents are positive by Eq.~\eqref{eq:delta-choice}; the explicit
bounds above also account for the floor in $t$. Absorbing the finite
prefactor into the exponent for large $n$ proves
Eq.~\eqref{eq:antidegradable-bound}, for example with any
\begin{equation}
  0<\gamma_r<\min\left\{\frac{a\delta^2\log e}{4},\
  \frac12\left[r-\delta\log d-\tfrac12h_2(\delta)\right]\right\}.
\end{equation}
Since $r>0$ was arbitrary, every positive rate is an exponential
strong-converse rate and $Q^{\exp(\dagger)}(\mathcal N)=0$.
\end{proof}

\begin{remark}
The bound is uniform over encoders and decoders and uses no structure of the
code; in particular it applies to non-stabilizer codes and to input states
whose reference marginal is far from maximally mixed. The only channel
parameter entering is the output dimension $|B|$.
\end{remark}

% =====================================================================
\section{Degradable channels}
\label{sec:degradable}
% =====================================================================

Morgan and Winter reduced the strong converse for degradable channels to a
strong converse for their self-complementary (``symmetric'') zero-capacity
channel~\cite{MorganWinter2014}. This associated channel is antidegradable,
so Theorem~\ref{thm:antidegradable} supplies the missing bound. However, the
reduction is quantitative, and closing it requires tracking how the error
parameter degrades.

\begin{theorem}\label{thm:degradable}
Let $\mathcal N$ be a finite-dimensional degradable channel. For every
$\Delta>0$ there is a $\gamma_\Delta>0$ such that every generalized
$(n,M_n)$ Bell-target scheme with
$r_n\ge Q(\mathcal N)+\Delta$ obeys $F_n\le2^{-\gamma_\Delta n}$ for all
sufficiently large $n$. Consequently
\begin{equation}
  Q^{\exp(\dagger)}(\mathcal N)=Q^\dagger(\mathcal N)
  =Q(\mathcal N)=Q^{(1)}(\mathcal N).
  \label{eq:degradable-capacity}
\end{equation}
\end{theorem}

\begin{proof}
Theorem~19 of Ref.~\cite{MorganWinter2014} provides a finite-dimensional
self-complementary channel $\mathcal M$, with output dimension
$d_{\mathcal M}$, and a constant $\mu>0$ such that, for every
$0\le\varepsilon<1$,
\begin{align}
  \log N_{\mathrm E}(n,\varepsilon\mid\mathcal N)
  \le\;& nQ^{(1)}(\mathcal N)+\delta_n(\lambda)
  \nonumber\\
  &+\log N_{\mathrm E}(n,1-\lambda\mid\mathcal M),
  \label{eq:mw-reduction}
\end{align}
where $\lambda=(1-\varepsilon)/5$ and
\begin{equation}
  \delta_n(\lambda):=
  \mu\sqrt{n\ln\!\left(\frac{64n|A'|^2}{\lambda^2}\right)}
  +8\log\frac1\lambda+O(\log n).
  \label{eq:mw-correction}
\end{equation}
Let $K,c_0$ be the constants of Lemma~\ref{lem:symmetric-rate} below and define,
for $0<c\le c_0$,
\begin{equation}
  \Psi(c):=\mu\sqrt{2c\ln2}+10c
  +\tfrac12h_2(K\sqrt c)+K\sqrt c\log d_{\mathcal M}.
  \label{eq:Psi-def}
\end{equation}
Since $\Psi(c)\to0$ as $c\downarrow0$, choose $0<c\le c_0$ such that
$\Psi(c)<\Delta$.

Suppose (aiming for a contradiction), that there is an infinitely large set $\mathcal I \subset \mathbb{N}$ of block lengths $n$ for which there is a generalized scheme with
$r_n\ge Q^{(1)}(\mathcal N)+\Delta$ and $F_n\ge2^{-cn}$. Put
$\varepsilon_n=\sqrt{1-F_n}$. Then
\begin{equation}
  \lambda_n:=\frac{1-\varepsilon_n}{5}
  =\frac{F_n}{5(1+\sqrt{1-F_n})}
  \ge\frac{F_n}{10}\ge\frac{2^{-cn}}{10}.
  \label{eq:lambda-bound}
\end{equation}
For each $n\in\mathcal I$, the given scheme is counted by
$N_{\mathrm E}(n,\varepsilon_n\mid\mathcal N)$, so we apply
Eq.~\eqref{eq:mw-reduction} at $\varepsilon=\varepsilon_n$. Along this
violating subsequence,
\begin{equation}
  \limsup_{\substack{n\to\infty\\ n\in\mathcal I}}
  \frac{\delta_n(\lambda_n)}{n}
  \le \mu\sqrt{2c\ln2}+8c,
\end{equation}
because $\log(1/\lambda_n)\le cn+\log10$. The proof of Lemma~\ref{lem:symmetric-rate} is pointwise, so it also
applies along the same subsequence and gives
\begin{align}
  \limsup_{\substack{n\to\infty\\ n\in\mathcal I}}
  \frac1n\log N_{\mathrm E}(n,1-\lambda_n\mid\mathcal M)
  &\le \tfrac12h_2(K\sqrt c)\nonumber\\
  &\quad+K\sqrt c\log d_{\mathcal M}+2c.
\end{align}
For every $n\in\mathcal I$, feasibility of the given scheme implies
\[
M_n\le
N_{\mathrm E}(n,\varepsilon_n\mid\mathcal N).
\]
Consequently, Eq.~\eqref{eq:mw-reduction} gives
\begin{align*}
Q^{(1)}(\mathcal N)+\Delta
\le & r_n
 =\frac1n\log M_n\\
\le &\frac1n
 \log N_{\mathrm E}(n,\varepsilon_n\mid\mathcal N)\\
\le & Q^{(1)}(\mathcal N)
 +\frac{\delta_n(\lambda_n)}{n}\\
 &+\frac1n\log
 N_{\mathrm E}(n,1-\lambda_n\mid\mathcal M).
\end{align*}
Taking the limsup along $n\in\mathcal I$ and using the two
estimates above therefore yields
\begin{equation}
Q^{(1)}(\mathcal N)+\Delta
\le Q^{(1)}(\mathcal N)+\Psi(c),
\end{equation}
contrary to $\Psi(c)<\Delta$. Hence, for all sufficiently large $n$, every
scheme at rate at least $Q^{(1)}(\mathcal N)+\Delta$ satisfies
$F_n<2^{-cn}$. Thus one may take $\gamma_\Delta=c$. Finally,
$Q^{(1)}(\mathcal N)=Q(\mathcal N)$ for degradable channels
\cite{DevetakShor2005}, and Eq.~\eqref{eq:capacity-hierarchy} closes the
chain.
\end{proof}
Above proof relies on the following Lemma:
\begin{lemma}\label{lem:symmetric-rate}
There are universal constants $K>0$ and $c_0>0$ with the following
property. Let $\mathcal M$ be antidegradable with output dimension 
$d_{\mathcal M}$, let $0<c\le c_0$, and let
$\lambda_n\ge2^{-cn}/10$. Then 
\begin{equation}
  \limsup_{n\to\infty}\frac1n\log N_{\mathrm E}(n,1-\lambda_n\mid\mathcal M)
  \le g_{\mathcal M}(c),
  \label{eq:symmetric-rate-bound}
\end{equation}
where
$g_{\mathcal M}(c):=\frac12h_2(K\sqrt c)+K\sqrt c\log d_{\mathcal M}+2c$
and $g_{\mathcal M}(c)\to0$ as $c\downarrow0$.
\end{lemma}
Here $N_E(n,\varepsilon\mid\mathcal M)$ denotes the largest
dimension $M$ for which a generalized $(n,M)$ Bell-target scheme
for $\mathcal M$ has purified-distance error at most $\varepsilon$. The proof can be found in the Supplemental Material.

The quantum erasure channel is one of the few noisy channels whose quantum
capacity is known exactly~\cite{BennettDiVincenzoSmolin1997}. Nevertheless,
an all-code strong converse remained open: previous results covered
maximally entangled inputs~\cite{SharmaWarsi2013} or almost all
codes~\cite{WildeWinter2014}. Since the channel is degradable for
$p\leq1/2$ and antidegradable for $p\geq1/2$,
Theorems~\ref{thm:degradable} and~\ref{thm:antidegradable} close this gap
and yield an all-code exponential strong converse throughout the full
parameter range.
\begin{corollary}[Erasure channel]\label{cor:erasure}
Let
\begin{equation}
  \mathcal E_{p,d}(\rho):=(1-p)\rho\oplus p\,\trace(\rho)\ketbra ee
  \label{eq:erasure-channel}
\end{equation}
be the $d$-dimensional erasure channel, with $\ket e$ an erasure flag
orthogonal to the input space. Then for every $p\in[0,1]$,
\begin{equation}
  Q^{\exp(\dagger)}(\mathcal E_{p,d})=Q(\mathcal E_{p,d})
  =\max\{1-2p,0\}\log d.
  \label{eq:erasure-strong-converse}
\end{equation}
\end{corollary}

\begin{proof}
The complementary channel of $\mathcal E_{p,d}$ is unitarily equivalent to
$\mathcal E_{1-p,d}$, so the channel is degradable for $p\le1/2$ and
antidegradable for $p\ge1/2$~\cite{BennettDiVincenzoSmolin1997}. Apply
Theorem~\ref{thm:degradable} in the first case and
Theorem~\ref{thm:antidegradable} in the second.
\end{proof}

This upgrades the almost-all-codes strong converse of
Ref.~\cite{WildeWinter2014} to a statement about every code, and strengthens
it from vanishing fidelity to exponentially vanishing fidelity.

% =====================================================================
\section{Post-processing and arbitrary channels}
\label{sec:postprocessing}
% =====================================================================
If two channels are concatenated, we find the following:
\begin{proposition}\label{prop:postprocessing}
Suppose $\mathcal N=\mathcal R\circ\widehat{\mathcal N}$ for channels
$\widehat{\mathcal N}:A'\to\widehat B$ and $\mathcal R:\widehat B\to B$.
Then every exponential strong-converse rate for $\widehat{\mathcal N}$ is
one for $\mathcal N$; in particular
$Q^{\exp(\dagger)}(\mathcal N)\le Q^{\exp(\dagger)}(\widehat{\mathcal N})$.
\end{proposition}

\begin{proof}
Given a scheme $(\rho_{RA'^{\,n}},\mathcal D_n)$ for $\mathcal N$, define
$\widehat{\mathcal D}_n:=\mathcal D_n\circ\mathcal R^{\otimes n}$. Since
$\mathcal N^{\otimes n}=\mathcal R^{\otimes n}\circ\widehat{\mathcal N}^{\otimes n}$,
the two schemes produce identical output states and hence identical
fidelities at identical rates.
\end{proof}

The point is that a channel need be neither degradable nor antidegradable to
inherit an exponential bound: it suffices that it can be decomposed into a sequence of channels, where one of them
is. Combined with Theorem~\ref{thm:degradable}, any degradable
$\widehat{\mathcal N}$ with $\mathcal N=\mathcal R\circ\widehat{\mathcal N}$
makes $Q^{(1)}(\widehat{\mathcal N})$ an all-code exponential
strong-converse rate for $\mathcal N$. This observation can be used to promote known capacity upper bounds into strong-converse bounds at no cost.

\subsection{The antidegradable weight}

To exploit Proposition~\ref{prop:postprocessing} systematically we quantify
the largest antidegradable part of a channel, following the squeezing
construction of Ref.~\cite{ZhuZhuWang2024}:
\begin{equation}
  \begin{aligned}
    w_{\mathrm{AD}}(\mathcal N):=\max_{w,\mathcal A,\mathcal S}\quad & w\\
    \text{subject to}\quad & \mathcal N=w\mathcal A+(1-w)\mathcal S,\\
    & 0\le w\le1,\ \mathcal A\ \text{antidegradable},\\
    & \mathcal S\ \text{a channel}.
  \end{aligned}
  \label{eq:antidegradable-weight}
\end{equation}
The value $w=0$ is always feasible and the maximum is attained in finite
dimensions.

We can use the antidegradable weight to upper bound the strong-converse rate:
\begin{theorem}\label{thm:universal}
For every finite-dimensional channel $\mathcal N:A'\to B$, with
$q:=1-w_{\mathrm{AD}}(\mathcal N)$,
\begin{equation}
  Q^{\exp(\dagger)}(\mathcal N)\le
  U_{\mathrm{AD}}(\mathcal N):=q\log\min\{|A'|,|B|\}.
  \label{eq:universal-bound}
\end{equation}
\end{theorem}

In order to prove the theorem, we make use of the following bound for flagged channels, which is proven in the Supplemental Material:
\begin{proposition}[Flagged bound]\label{prop:flagged}
Let $\mathcal K_0:A'\to B^{(0)}$ be an arbitrary channel, let
$\mathcal K_1:A'\to B^{(1)}$ be antidegradable, put $d_0:=|B^{(0)}|$, and let
$q+w=1$ with $q,w\ge0$. Then the flagged channel
\begin{equation}
  \mathcal F(\rho):=q\,\mathcal K_0(\rho)\otimes\ketbra00_C
  +w\,\mathcal K_1(\rho)\otimes\ketbra11_C
  \label{eq:flagged-channel}
\end{equation}
satisfies $Q^{\exp(\dagger)}(\mathcal F)\le q\log d_0$.
\end{proposition}
If $B^{(0)}$ and $B^{(1)}$ are different, we regard both as subspaces of
$B^{(0)}\oplus B^{(1)}$ and omit the corresponding embeddings.

This proposition allows to prove the theorem.
\begin{proof}[Proof of Theorem~\ref{thm:universal}]
Fix an optimal decomposition $\mathcal N=q\mathcal S+w\mathcal A$ with
$\mathcal A$ antidegradable, and apply
Proposition~\ref{prop:flagged} twice.

We first take
$\mathcal K_0=\mathcal S$ and $\mathcal K_1=\mathcal A$ and write $\mathcal{N} = \mathcal{E} \circ \mathcal{F}$ with $\mathcal{F}$ from Proposition~\ref{prop:flagged} and $\mathcal{E}$ simply erasing the flag. Thus, Propositions~\ref{prop:postprocessing} and \ref{prop:flagged} show $Q^{\exp(\dagger)}(\mathcal N) \leq Q^{\exp(\dagger)}(\mathcal F) \leq q\log |B|$.

Now take instead $\mathcal K_0'=\operatorname{id}_{A'}$ and
$\mathcal K_1'=\mathcal A$ and define the post-processing $\mathcal{E}'(\sigma \otimes \ketbra{c}{c})=\delta_{c,0}S(\sigma) + \delta_{c,1}\sigma$, i.e.,  it reads the flag and applies $\mathcal S$ if the flag is zero and the identity otherwise, while discarding the flag. Therefore,  $Q^{\exp(\dagger)}(\mathcal N) \leq Q^{\exp(\dagger)}(\mathcal F') \leq q\log |A'|$. Taking the smaller of the two gives
Eq.~\eqref{eq:universal-bound}.
\end{proof}

\begin{remark}
The second choice replaces the residual channel by the identity, which can
only help the code and is therefore legitimate in a converse. What is gained
is that the systems shared between two decoder placements at the residual
positions are now channel \emph{inputs} rather than outputs, which is
exactly what turns $\log|B|$ into $\log|A'|$. Proceeding this way also
avoids having to exhibit a degradable parent channel for $\mathcal N$, which
would require knowing that $\mathcal A$ is a degradation of a
self-complementary channel.
\end{remark}

\subsection{Semidefinite formulation}

Because antidegradability is equivalent to the existence of a symmetric
two-output extension, Eq.~\eqref{eq:antidegradable-weight} is a semidefinite
program. Let $d_A:=|A'|$, let $A\simeq A'$, and let
$\Phi_{AA'}$ denote the normalized maximally entangled state of
Eq.~\eqref{eq:mes}, with $M=d_A$, on $A\otimes A'$. Define the
unnormalized Choi operator
\begin{equation}
  J_{\mathcal N}
  :=d_A\left(\operatorname{id}_A\otimes\mathcal N\right)(\Phi_{AA'}),
\end{equation}
so that
\[
\trace_B J_{\mathcal N}=I_A.
\]
With $B'\simeq B$ and $F_{BB'}$ the swap,
$w_{\mathrm{AD}}(\mathcal N)$ is the optimum over
$w\in\mathbb R$, $Y_{AB}=Y_{AB}^{\dagger}$, and
$X_{ABB'}=X_{ABB'}^{\dagger}$ of
\begin{align}
  \text{maximize}\quad & w \nonumber\\
  \text{subject to}\quad & 0\le w\le1,\quad 0\le Y_{AB}\le J_{\mathcal N},\nonumber\\
  & \trace_B Y_{AB}=w I_A,\nonumber\\
  & X_{ABB'}\ge0,\quad \trace_{B'}X_{ABB'}=Y_{AB},\nonumber\\
  & F_{BB'}X_{ABB'}F_{BB'}=X_{ABB'} .
  \label{eq:weight-sdp}
\end{align}
Feasibility in both directions is direct: from a decomposition one takes
$Y_{AB}=wJ_{\mathcal A}$ and $X_{ABB'}$ is equal to $w$ times the Choi operator of a
symmetric extension; conversely, for $w>0$ the operator $Y_{AB}/w$ is the
Choi operator of a channel, $X_{ABB'}/w$ exhibits its symmetric extension,
and for $w<1$ the operator $(J_{\mathcal N}-Y_{AB})/(1-w)$ is the Choi
operator of the residual channel. The bound $w\le1$ also follows from
$Y_{AB}\le J_{\mathcal N}$ and the two trace constraints, but it is displayed for
clarity. 

Every channel with $J_{\mathcal N}>0$  has $w_{\mathrm{AD}}(\mathcal N)>0$.
Indeed, the replacer channel
$\mathcal A_\tau(\rho)=\trace(\rho)\tau_B$ with full-rank $\tau_B$ is
antidegradable and has Choi operator $I_A\otimes\tau_B$. An explicit
feasible choice is any
\begin{equation}
  0<w\le\min\left\{1,\lambda_{\min}\!\left[
  (I_A\otimes\tau_B^{-1/2})J_{\mathcal N}
  (I_A\otimes\tau_B^{-1/2})\right]\right\},
\end{equation}
for which $J_{\mathcal N}-w(I_A\otimes\tau_B)\ge0$. Hence
$U_{\mathrm{AD}}(\mathcal N)<\log\min\{|A'|,|B|\}$ strictly, for every
full-Choi-rank channel. We stress that $U_{\mathrm{AD}}$ is an upper bound
on the exponential strong-converse capacity and need not equal
$Q(\mathcal N)$.

% =====================================================================
\section{Pauli channels}
\label{sec:qubit}
% =====================================================================

For a qubit Pauli channel
$\Lambda_{\bm p}(\rho)=\sum_{i=0}^{3}p_i\sigma_i\rho\sigma_i$ with
$p_0\ge p_i>0$ for $i=1,2,3$, set
\begin{equation}
  \beta(\bm p):=\sum_{1\le i<j\le3}\left(\sqrt{p_i}+\sqrt{p_j}\right)^2 .
  \label{eq:pauli-beta}
\end{equation}
Ref.~\cite{ZhuZhuWang2024} proves the following: either
$\Lambda_{\bm p}$ is antidegradable, or it is non-antidegradable and
$w_{\mathrm{AD}}(\Lambda_{\bm p})=\beta(\bm p)$. Combining that result with
Theorems~\ref{thm:antidegradable} and~\ref{thm:universal} gives an upper bound on $Q^{\exp(\dagger)}$.
\begin{equation}
 Q^{\exp(\dagger)}(\Lambda_{\bm p})\leq U_{\mathrm P}(\bm p)=
  \begin{cases}
    0, & \Lambda_{\bm p}\ \text{antidegradable},\\
    1-\beta(\bm p), & \Lambda_{\bm p}\ \text{not antidegradable}.
  \end{cases}
  \label{eq:pauli-rate}
\end{equation}
The cited analytic formula assumes strictly positive probabilities; boundary
Pauli channels remain covered by the SDP in Eq.~\eqref{eq:weight-sdp}, and
we make no additional closed-form claim for them.
For the depolarizing channel
$\mathcal P_p(\rho)=(1-p)\rho+\frac p3\sum_{i=1}^{3}\sigma_i\rho\sigma_i$
with $0\le p\le3/4$ this reduces to 
\begin{equation}
  U_{\mathrm P}(p)=\max\{1-4p,0\},
  \label{eq:depolarizing-rate}
\end{equation}
where $p=0$ is the identity channel and the depolarizing channel is
antidegradable precisely for $p\ge1/4$.

\begin{figure}[t]
    \centering
    \includegraphics[width=\columnwidth]{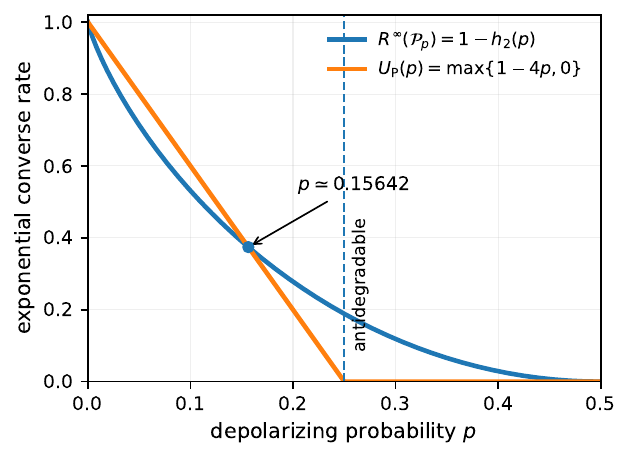}
    \caption{\textbf{Exponential strong-converse bounds for the qubit
    depolarizing channel.} For
    $\mathcal P_{p}(\rho)=(1-p)\rho+\frac p3\sum_{i=1}^{3}\sigma_i\rho\sigma_i$,
    our Pauli-channel bound is $U_{\mathrm P}(p)=\max\{1-4p,0\}$, while
    the regularized channel Rains information is the piecewise function in
    Eq.~\eqref{eq:depolarizing-rains}.
    The two cross at $p\simeq0.15642$, above which our bound is strictly
    tighter. Our bound vanishes at the exact antidegradability threshold
    $p=1/4$, whereas the Rains bound stays strictly positive there.}
    \label{fig:depolarizing-rains-comparison}
\end{figure}

For comparison, define the regularized channel Rains information by
$R^\infty(\mathcal N):=\inf_{k\geq 1}
k^{-1}R(\mathcal N^{\otimes k})$,
where
$R(\mathcal N):=\max_{\psi_{AA'}}R(A{:}B)_{(\operatorname{id}_A\otimes
\mathcal N)(\psi)}$, with
\begin{equation}
  R(A{:}B)_\rho:=\inf_{\substack{\tau_{AB}\ge0\\
  \lVert\tau_{AB}^{T_B}\rVert_1\le1}}D(\rho_{AB}\Vert\tau_{AB})
  \label{eq:rains-state}
\end{equation}
the Rains relative entropy \cite{Rains1999}. Let
$\tilde{J}_{\mathcal P_p}:=(\operatorname{id}\otimes\mathcal P_p)(\Phi_2)$ be the
normalized Choi state. Since $\mathcal P_p$ is a Pauli channel, standard
teleportation gives an LOCC simulation of the channel from
$\tilde{J}_{\mathcal P_p}$~\cite{CopeEtAl2017}. For any input $\psi_{AA'}$, the
state $(\operatorname{id}_A\otimes\mathcal P_p)(\psi_{AA'})$ is therefore obtained
by LOCC from the locally prepared input $\psi_{AA'}$ and the resource state
$\tilde{J}_{\mathcal P_p}$. Monotonicity of the Rains relative entropy under LOCC
gives $R(\mathcal P_p)\le R(A:B)_{\tilde{J}_{\mathcal{P}_p}}$, while the maximally
entangled input gives the reverse inequality. Thus
$R(\mathcal P_p)=R(A:B)_{\tilde{J}_{\mathcal{P}_p}}$. The same simulation argument applies
to every tensor power, and the Rains expression for this isotropic-state
family is additive~\cite{Rains1999}. Consequently
\begin{equation}
  R^\infty(\mathcal P_p)=
  \begin{cases}
    1-h_2(p), & 0\le p\le\tfrac12,\\
    0, & \tfrac12\le p\le\tfrac34,
  \end{cases}
  \label{eq:depolarizing-rains}
\end{equation}
The role of the regularized channel Rains information as an upper bound on exponential
strong-converse rate follows from
Ref.~\cite{TomamichelWildeWinter2017}. As
Fig.~\ref{fig:depolarizing-rains-comparison} shows, our bound crosses
$R^\infty$ at $p\simeq0.15642$ and is strictly tighter above; more
tellingly, it vanishes exactly at the antidegradability threshold $p=1/4$,
where the Rains bound remains strictly positive. Since
the depolarizing capacity is unknown for $p<1/4$, this is an upper bound on
the exponential strong-converse capacity, not a determination of the
capacity.

% =====================================================================
\section{A nondegradable family at its exact capacity}
\label{sec:amplitude-damping}
% =====================================================================

Proposition~\ref{prop:postprocessing} can be tight. Let $d\ge3$, fix an
orthonormal basis $\{\ket0,\dots,\ket{d-1}\}$, and for $\eta\in[0,1]$ define
the multilevel amplitude-damping channel
$\Omega^{[\eta]}$ by Kraus operators
\begin{align}
  K_0^{[\eta]}&:=\ketbra00+\sqrt{1-\eta}\,\ketbra11
  +\sum_{k=2}^{d-1}\ketbra kk,
  \nonumber\\
  K_1^{[\eta]}&:=\sqrt\eta\,\ketbra01 ,
  \label{eq:ad-kraus}
\end{align}
which damps $\ket1\to\ket0$ with probability $\eta$ and acts as the identity
on $\mathcal H_{\bar1}:=\mathrm{span}\{\ket0,\ket2,\dots,\ket{d-1}\}$; this family
was studied in Ref.~\cite{ChessaGiovannetti2021}, while recent work analyzes
more general finite-dimensional lossy and multilevel amplitude-damping
channels~\cite{CocciarettoGiovannetti2026}. Because the population
remaining in $\ket1$ multiplies,
\begin{equation}
  \Omega^{[\eta_1]}\circ\Omega^{[\eta_2]}
  =\Omega^{[\eta_1+\eta_2-\eta_1\eta_2]},
  \label{eq:ad-composition}
\end{equation}
so that for every $\eta\in[1/2,1]$,
\begin{equation}
  \Omega^{[\eta]}=\Omega^{[2\eta-1]}\circ\Omega^{[1/2]}.
  \label{eq:ad-factorization}
\end{equation}

\begin{theorem}\label{thm:amplitude-damping}
For every $d\ge3$ and every $\eta\in[1/2,1]$,
\begin{equation}
  Q^{\exp(\dagger)}\!\left(\Omega^{[\eta]}\right)
  =Q\!\left(\Omega^{[\eta]}\right)=\log(d-1),
  \label{eq:ad-capacity}
\end{equation}
and for $\eta>1/2$ the channel is neither degradable nor antidegradable.
\end{theorem}

\begin{proof}
We first prove the statement for $\eta=1/2$ in arbitrary dimension. Use
the Stinespring isometry
$V_\eta=K_0^{[\eta]}\otimes\ket0_E+K_1^{[\eta]}\otimes\ket1_E$.
For $0\le\eta\le1/2$, put $\gamma=\eta/(1-\eta)$ and define a channel
$\mathcal D_\eta:B\to E$ by the Kraus operators
\begin{align}
  L_0&:=\ket0_E\!\bra0_B+\sqrt\gamma\,\ket1_E\!\bra1_B,\nonumber\\
  L_1&:=\sqrt{1-\gamma}\,\ket0_E\!\bra1_B,\nonumber\\
  L_k&:=\ket0_E\!\bra k_B,\qquad k=2,\ldots,d-1.
  \label{eq:ad-degrading-kraus}
\end{align}
They satisfy $\sum_kL_k^\dagger L_k=I_B$, and direct
substitution gives
\[
\Omega^{[\eta]c}=\mathcal D_\eta\circ\Omega^{[\eta]}.
\]
Thus $\Omega^{[\eta]}$ is degradable for every
$0\leq\eta\leq1/2$; in particular, $\Omega^{[1/2]}$ is
degradable.

It remains to maximize its coherent information. The channel is covariant
under diagonal phase unitaries. Since coherent information is concave in the
input for a degradable channel, averaging over those phases cannot decrease
it, so an optimizer may be taken diagonal,
$\rho=\sum_{k=0}^{d-1}p_k\ketbra kk$. Set $q=p_1/2$. The receiver and
environment eigenvalue lists are
\begin{equation}
  \bigl(p_0+q,q,p_2,\ldots,p_{d-1}\bigr),
  \qquad (1-q,q),
\end{equation}
respectively. If
$\widetilde{\bm p}:=(p_0+q,p_2,\ldots,p_{d-1})/(1-q)$, then
\begin{align}
  \Ic(\rho,\Omega^{[1/2]})
  &=H(B)-H(E)\nonumber\\
  &=(1-q)H(\widetilde{\bm p})\nonumber\\
  &\le(1-q)\log(d-1)\le\log(d-1).
  \label{eq:ad-coherent-bound}
\end{align}
Equality is attained for $p_1=0$ and the maximally mixed state on
$\mathcal H_{\bar1}$. Hence,  and since $\Omega^{[1/2]}$ is degradable,
$Q(\Omega^{[1/2]})=Q^{(1)}(\Omega^{[1/2]})=\log(d-1)$.
Equation~\eqref{eq:ad-factorization},
Proposition~\ref{prop:postprocessing}, and Theorem~\ref{thm:degradable} now
give
$Q^{\exp(\dagger)}(\Omega^{[\eta]})\le\log(d-1)$. Conversely,
$\Omega^{[\eta]}$ and $\Omega^{[1/2]}$ act identically on $\mathcal H_{\bar1}$, which achieves
$\log(d-1)$, and Eq.~\eqref{eq:capacity-hierarchy} closes the sandwich.

Finally let $\eta>1/2$. If $\Omega^{[\eta]}$ were degradable, restricting
its input and receiver output to $\mathrm{span}\{\ket0,\ket1\}$ would make
the qubit amplitude-damping channel with parameter $\eta$ degradable, which is known to be not degradable for $\eta>1/2$ \cite{Cubitt2008Degradable}. $\Omega^{[\eta]}$ is not antidegradable either, because the positive rate $\log(d-1)$ is achievable.
\end{proof}

Thus, the multilevel amplitude-damping channel $\Omega^{[\eta]}$ for $\eta > 1/2$ constitutes a family of neither degradable nor antidegradable channels that nevertheless obeys an all-code
exponential strong converse at its exact quantum capacity.

% =====================================================================
\section{Discussion}
\label{sec:discussion}
% =====================================================================

Our results show that
the quantum capacity of a degradable channel is a sharp threshold, and
crossing it is exponentially costly. Tolerating a fixed error, however
large, buys no rate at all, and there is no intermediate regime in which a
scheme limps along above capacity at constant fidelity. The same holds
for antidegradable channels, where
the threshold is zero. Together the two results cover the erasure channel
completely and give, to our knowledge, the first all-code strong converse for it.

The mechanism is worth isolating from the application. No-cloning is
usually invoked qualitatively, to argue that some capacity vanishes. What
Sec.~\ref{sec:antidegradable} shows is that it also has quantitative,
finite-block content: a swap-symmetric extension does not merely forbid two
receivers from both decoding, it forbids $2^n$ overlapping receivers from
decoding with any but exponentially small probability. The proof techniques used here can probably be applied wherever a resource admits many
overlapping simulations of the same measurement. Finite-block bounds in
quantum error correction and converses for other communication tasks, are natural candidates for such extensions.

Two limitations are worth naming. First, the constants are not optimized:
Proposition~\ref{prop:finite-block} gives an exponent of order $\delta^2$,
and we make no claim that the resulting $\gamma_r$ is close to the true
strong-converse exponent. Determining the optimal exponent, even for the
erasure channel, is open. Second, for a general channel the bound
$U_{\mathrm{AD}}$ of Theorem~\ref{thm:universal} does not meet the capacity,
and the gap in Fig.~\ref{fig:depolarizing-rains-comparison} below
$p\simeq0.156$ displays this discrepancy. Better degradable extensions, additional structure
in the residual channel $\mathcal S$, and combining the antidegradable
weight with Rains-type relaxations are all plausible routes to tightening
it.

\begin{acknowledgments}
 Some technical steps of the proofs presented here were developed with the help of
ChatGPT (OpenAI, GPT-5.6 Sol; accessed July 2026). The initial draft of the paper was written by hand and iteratively improved with the help of ChatGPT 5.6 Sol and Claude Opus 4.8, both accessed August 2026. All AI-assisted material was
reviewed and revised by the authors, who take full responsibility for the
content. This work has been supported by the Federal
Ministry of Research, Technology and Space (BMFTR Projects QR.N, Grant
No.~16KIS2202 and QSolid, Grant No.~13N16163). We also acknowledge
financial support by Deutsche Forschungsgemeinschaft (DFG, German Research
Foundation) under Germany's Excellence Strategy -- Cluster of Excellence
Matter and Light for Quantum Computing (ML4Q) EXC 2004/1 -- 390534769.
\end{acknowledgments}

\bibliography{channels}

@article{Lloyd1997,
  author  = {Lloyd, Seth},
  title   = {Capacity of the Noisy Quantum Channel},
  journal = {Physical Review A},
  volume  = {55},
  number  = {3},
  pages   = {1613--1622},
  year    = {1997},
  doi     = {10.1103/PhysRevA.55.1613}
}

@article{Devetak2005,
  author  = {Devetak, Igor},
  title   = {The Private Classical Capacity and Quantum Capacity of a Quantum Channel},
  journal = {IEEE Transactions on Information Theory},
  volume  = {51},
  number  = {1},
  pages   = {44--55},
  year    = {2005},
  doi     = {10.1109/TIT.2004.839515},
  eprint  = {quant-ph/0304127},
  archivePrefix = {arXiv}
}

@article{KretschmannWerner2004,
  author  = {Kretschmann, Dennis and Werner, Reinhard F.},
  title   = {Tema Con Variazioni: Quantum Channel Capacity},
  journal = {New Journal of Physics},
  volume  = {6},
  pages   = {26},
  year    = {2004},
  doi     = {10.1088/1367-2630/6/1/026},
  eprint  = {quant-ph/0311037},
  archivePrefix = {arXiv}
}

@article{SharmaWarsi2013,
  author  = {Sharma, Naresh and Warsi, Naqueeb Ahmad},
  title   = {Fundamental Bound on the Reliability of Quantum Information Transmission},
  journal = {Physical Review Letters},
  volume  = {110},
  number  = {8},
  pages   = {080501},
  year    = {2013},
  doi     = {10.1103/PhysRevLett.110.080501},
  eprint  = {1205.1712},
  archivePrefix = {arXiv}
}

@article{KaurEtAl2021,
  author  = {Kaur, Eneet and Das, Siddhartha and Wilde, Mark M. and Winter, Andreas},
  title   = {Resource Theory of Unextendibility and Non-Asymptotic Quantum Capacity},
  journal = {Physical Review A},
  volume  = {104},
  number  = {2},
  pages   = {022401},
  year    = {2021},
  doi     = {10.1103/PhysRevA.104.022401},
  eprint  = {1803.10710},
  archivePrefix = {arXiv}
}

@article{MorganWinter2014,
  author  = {Morgan, Ciara and Winter, Andreas},
  title   = {{``Pretty Strong'' Converse for the Quantum Capacity of Degradable Channels}},
  journal = {IEEE Transactions on Information Theory},
  volume  = {60},
  number  = {1},
  pages   = {317--333},
  year    = {2014},
  doi     = {10.1109/TIT.2013.2288971},
  eprint  = {1301.4927},
  archivePrefix = {arXiv},
  primaryClass = {quant-ph}
}

@inproceedings{WildeWinter2014,
  author    = {Wilde, Mark M. and Winter, Andreas},
  title     = {Strong Converse for the Quantum Capacity of the Erasure Channel for Almost All Codes},
  booktitle = {Proceedings of the 9th Conference on the Theory of Quantum Computation, Communication and Cryptography},
  series    = {Leibniz International Proceedings in Informatics},
  volume    = {27},
  pages     = {52--66},
  year      = {2014},
  doi       = {10.4230/LIPIcs.TQC.2014.52},
  eprint    = {1402.3626},
  archivePrefix = {arXiv},
  primaryClass = {quant-ph}
}

@article{TomamichelWildeWinter2017,
  author  = {Tomamichel, Marco and Wilde, Mark M. and Winter, Andreas},
  title   = {Strong Converse Rates for Quantum Communication},
  journal = {IEEE Transactions on Information Theory},
  volume  = {63},
  number  = {1},
  pages   = {715--727},
  year    = {2017},
  doi     = {10.1109/TIT.2016.2615847},
  eprint  = {1406.2946},
  archivePrefix = {arXiv},
  primaryClass = {quant-ph}
}

@article{KhanianHirche2025,
  author  = {Baghali Khanian, Zahra and Hirche, Christoph},
  title   = {On Strong Converse Bounds for the Private and Quantum Capacities of Antidegradable Channels},
  journal = {arXiv preprint arXiv:2507.15661},
  year    = {2025},
  eprint  = {2507.15661},
  archivePrefix = {arXiv},
  primaryClass = {quant-ph}
}

@article{Tomamichel2026,
  author  = {Tomamichel, Marco},
  title   = {A Strong Converse for Stabilizer Codes over {Pauli} Channels via the Blowing-Up Lemma},
  journal = {arXiv preprint arXiv:2607.23450},
  year    = {2026},
  eprint  = {2607.23450},
  archivePrefix = {arXiv},
  primaryClass = {quant-ph}
}

@article{Stinespring1955,
  author  = {Stinespring, W. Forrest},
  title   = {Positive Functions on {$C^*$}-Algebras},
  journal = {Proceedings of the American Mathematical Society},
  volume  = {6},
  number  = {2},
  pages   = {211--216},
  year    = {1955},
  doi     = {10.2307/2032342}
}

@article{DevetakShor2005,
  author  = {Devetak, Igor and Shor, Peter W.},
  title   = {The Capacity of a Quantum Channel for Simultaneous Transmission
             of Classical and Quantum Information},
  journal = {Communications in Mathematical Physics},
  volume  = {256},
  number  = {2},
  pages   = {287--303},
  year    = {2005},
  doi     = {10.1007/s00220-005-1317-6},
  eprint  = {quant-ph/0311131},
  archivePrefix = {arXiv}
}

@article{BennettDiVincenzoSmolin1997,
  author  = {Bennett, Charles H. and DiVincenzo, David P. and Smolin, John A.},
  title   = {Capacities of Quantum Erasure Channels},
  journal = {Physical Review Letters},
  volume  = {78},
  number  = {16},
  pages   = {3217--3220},
  year    = {1997},
  doi     = {10.1103/PhysRevLett.78.3217},
  eprint  = {quant-ph/9701015},
  archivePrefix = {arXiv}
}

@article{ZhuZhuWang2024,
  author       = {Zhu, Chengkai and Zhu, Chenghong and Wang, Xin},
  title        = {Estimate Distillable Entanglement and Quantum Capacity by
                  Squeezing Useless Entanglement},
  journal      = {IEEE Journal on Selected Areas in Communications},
  volume       = {42},
  number       = {7},
  pages        = {1850--1860},
  year         = {2024},
  month        = jul,
  doi          = {10.1109/JSAC.2024.3380081},
  eprint       = {2303.07228},
  archiveprefix = {arXiv},
  primaryclass = {quant-ph}
}

@article{ChessaGiovannetti2021,
  author       = {Chessa, Stefano and Giovannetti, Vittorio},
  title        = {Quantum Capacity Analysis of Multi-Level Amplitude Damping
                  Channels},
  journal      = {Communications Physics},
  volume       = {4},
  pages        = {22},
  year         = {2021},
  doi          = {10.1038/s42005-021-00524-4},
  eprint       = {2008.00477},
  archiveprefix = {arXiv},
  primaryclass = {quant-ph}
}

@article{Bennett1993Teleportation,
  author  = {Bennett, Charles H. and Brassard, Gilles and Cr{\'e}peau, Claude
             and Jozsa, Richard and Peres, Asher and Wootters, William K.},
  title   = {Teleporting an Unknown Quantum State via Dual Classical and
             {E}instein-{P}odolsky-{R}osen Channels},
  journal = {Physical Review Letters},
  volume  = {70},
  number  = {13},
  pages   = {1895--1899},
  year    = {1993},
  doi     = {10.1103/PhysRevLett.70.1895}
}

@article{Cirac1999Distributed,
  author       = {Cirac, J. Ignacio and Ekert, Artur K. and Huelga, Susana F.
                  and Macchiavello, Chiara},
  title        = {Distributed Quantum Computation over Noisy Channels},
  journal      = {Physical Review A},
  volume       = {59},
  number       = {6},
  pages        = {4249--4254},
  year         = {1999},
  doi          = {10.1103/PhysRevA.59.4249},
  eprint       = {quant-ph/9803017},
  archiveprefix = {arXiv}
}

@article{Wehner2018QuantumInternet,
  author  = {Wehner, Stephanie and Elkouss, David and Hanson, Ronald},
  title   = {Quantum Internet: A Vision for the Road Ahead},
  journal = {Science},
  volume  = {362},
  number  = {6412},
  pages   = {eaam9288},
  year    = {2018},
  doi     = {10.1126/science.aam9288}
}

@article{Shannon1948,
  author  = {Shannon, Claude E.},
  title   = {A Mathematical Theory of Communication},
  journal = {The Bell System Technical Journal},
  volume  = {27},
  pages   = {379--423, 623--656},
  year    = {1948},
  doi= {doi:10.1002/j.1538-7305.1948.tb01338.x}
}

@article{WolfPerezGarcia2007,
  author       = {Wolf, Michael M. and P{\'e}rez-Garc{\'i}a, David},
  title        = {Quantum Capacities of Channels with Small Environment},
  journal      = {Physical Review A},
  volume       = {75},
  number       = {1},
  pages        = {012303},
  year         = {2007},
  doi          = {10.1103/PhysRevA.75.012303},
  eprint       = {quant-ph/0607070},
  archiveprefix = {arXiv}
}

@article{MyhrLutkenhaus2009,
  author        = {Myhr, Geir Ove and L{\"u}tkenhaus, Norbert},
  title         = {Spectrum Conditions for Symmetric Extendible States},
  journal       = {Physical Review A},
  volume        = {79},
  number        = {6},
  pages         = {062307},
  year          = {2009},
  doi           = {10.1103/PhysRevA.79.062307},
  eprint        = {0812.3667},
  archiveprefix = {arXiv},
  primaryclass  = {quant-ph}
}

@article{deWolfNote,
  author        = {de Wolf, Ronald},
  title         = {A note on quantum algorithms and the minimal degree of $\epsilon$-error polynomials for symmetric functions},
  journal       = {Quantum Information \&  Computation},
  volume        = {8},
  number        = {10},
  pages         = {943},
  year          = {2008},
  doi           = {10.26421/QIC8.10-4}
}

@article{BarnumKnillNielsen2000,
  author        = {Barnum, Howard and Knill, Emanuel and Nielsen, Michael A.},
  title         = {On Quantum Fidelities and Channel Capacities},
  journal       = {IEEE Transactions on Information Theory},
  volume        = {46},
  number        = {4},
  pages         = {1317--1329},
  year          = {2000},
  doi           = {10.1109/18.850671},
  eprint        = {quant-ph/9809010},
  archivePrefix = {arXiv}
}

@article{CocciarettoGiovannetti2026,
  author        = {Cocciaretto, Sofia and Giovannetti, Vittorio},
  title         = {Quantum Capacity Analysis of Finite-Dimensional Lossy Channels},
  journal       = {arXiv preprint arXiv:2601.18960},
  year          = {2026},
  eprint        = {2601.18960},
  archivePrefix = {arXiv},
  primaryClass  = {quant-ph}
}

@article{CopeEtAl2017,
  author        = {Cope, Thomas P. W. and Hetzel, Leon and Banchi, Leonardo and Pirandola, Stefano},
  title         = {Simulation of Non-{Pauli} Channels},
  journal       = {Physical Review A},
  volume        = {96},
  number        = {2},
  pages         = {022323},
  year          = {2017},
  doi           = {10.1103/PhysRevA.96.022323},
  eprint        = {1706.05384},
  archiveprefix = {arXiv},
  primaryclass  = {quant-ph}
}

@article{Cubitt2008Degradable,
    author = {Cubitt, Toby S. and Ruskai, Mary Beth and Smith, Graeme},
    title = {The structure of degradable quantum channels},
    journal = {Journal of Mathematical Physics},
    volume = {49},
    number = {10},
    pages = {102104},
    year = {2008},
    month = {10},
    doi = {10.1063/1.2953685}
}

@article{Rains1999,
  author        = {Rains, E. M.},
  title         = {Bound on Distillable Entanglement},
  journal       = {Physical Review A},
  volume        = {60},
  number        = {1},
  pages         = {179--184},
  year          = {1999},
  doi           = {10.1103/PhysRevA.60.179},
  eprint        = {quant-ph/9809082},
  archivePrefix = {arXiv},
  primaryClass  = {quant-ph}
}

\clearpage
\onecolumngrid
\appendix

\setcounter{theorem}{0}
\renewcommand{\thetheorem}{S\arabic{theorem}}
\renewcommand{\theHtheorem}{S\arabic{theorem}}

\section*{Supplemental Material}

This appendix proves, in order: the finite-block bound for antidegradable
channels (Proposition~\ref{prop:finite-block} of the main text), via three
lemmas that isolate the operator-theoretic, the Boolean-analytic, and the
no-cloning content of the argument; the rate bound for antidegradable channels
at exponentially small fidelity (Lemma~\ref{lem:symmetric-rate}) and the flagged
bound (Proposition~\ref{prop:flagged}).

\setcounter{section}{0}
\renewcommand{\thesection}{S\arabic{section}}
\renewcommand{\theHsection}{supp.S\arabic{section}}
\makeatletter
\@removefromreset{equation}{section}
\makeatother
\renewcommand{\theequation}{S\arabic{equation}}
\renewcommand{\theHequation}{supp.S\arabic{equation}}
\setcounter{equation}{0}

% =====================================================================
\section{Setup and the common Hilbert space}
\label{sm:setup}
% =====================================================================

Throughout this section $\mathcal N:A'\to B$ is antidegradable,
$d:=|B|$, and $\widetilde{\mathcal N}:A'\to B_0B_1$ is a swap-symmetric
two-output extension, so that
$\trace_{B_1}\circ\widetilde{\mathcal N}=\trace_{B_0}\circ\widetilde{\mathcal N}=\mathcal N$
with $B_0\simeq B_1\simeq B$. We fix a generalized $(n,M)$ Bell-target
scheme $(\rho_{RA'^{\,n}},\mathcal D_n)$, with $|R|=|\widehat S|=M$ and
$\mathcal D_n:B^n\to\widehat S$.

Applying $\widetilde{\mathcal N}$ to each channel use produces
\begin{equation}
  \omega_{R\mathbf B}
  :=\left(\operatorname{id}_R\otimes\widetilde{\mathcal N}^{\otimes n}\right)
  (\rho_{RA'^{\,n}}),
  \qquad
  \mathbf B:=\bigotimes_{j=1}^{n}\left(B_{j,0}\otimes B_{j,1}\right).
  \label{eq:sm-output-state}
\end{equation}
For $x\in\{0,1\}^n$ put $B_x:=\bigotimes_{j=1}^{n}B_{j,x_j}$ and
\begin{equation}
  s(x,y):=\left|\{j:x_j=y_j\}\right| .
  \label{eq:sm-agreement}
\end{equation}
Because each marginal of $\widetilde{\mathcal N}$ equals $\mathcal N$,
\begin{equation}
  \trace_{\mathbf B\setminus B_x}\left[\omega_{R\mathbf B}\right]
  =\left(\operatorname{id}_R\otimes\mathcal N^{\otimes n}\right)
  (\rho_{RA'^{\,n}})
  \qquad\text{for every }x\in\{0,1\}^n .
  \label{eq:sm-marginal-invariance}
\end{equation}

To each $x$ we attach a private ancilla $G_x$, prepared in $\ket0$, and a
unitary dilation; the ancilla and environment are enlarged as necessary so
that the input and output dimensions match,
\begin{equation}
  U_x:B_xG_x\longrightarrow \widehat SE_x
\end{equation}
of a copy $\mathcal D_x$ of the decoder, so that
$\mathcal D_x(\tau)=\trace_{E_x}[U_x(\tau\otimes\ketbra00_{G_x})U_x^\dagger]$.
All objects are placed on the common space
\begin{equation}
  \mathcal H_{\mathrm{tot}}
  :=\mathcal H_R\otimes
  \bigotimes_{j=1}^{n}\left(\mathcal H_{B_{j,0}}\otimes\mathcal H_{B_{j,1}}\right)
  \otimes\bigotimes_{y\in\{0,1\}^{n}}\mathcal H_{G_y},
  \qquad
  \widetilde\omega:=\omega_{R\mathbf B}\otimes
  \bigotimes_{y\in\{0,1\}^n}\ketbra00_{G_y},
  \label{eq:sm-common-space}
\end{equation}
and the fidelity projector of the $x$-th placement is
\begin{equation}
  P_x:=\left(I_R\otimes U_x^\dagger\right)
  (\Phi_{R\widehat S}\otimes I_{E_x})
  \left(I_R\otimes U_x\right)
  \otimes I_{B_{\overline x}}\otimes I_{G_{\neq x}},
  \label{eq:sm-projector}
\end{equation}
where
$B_{\overline x}:=\bigotimes_{j=1}^{n}B_{j,1-x_j}$ is the complementary
selection and $G_{\neq x}:=\bigotimes_{y\neq x}G_y$. Canonical tensor-factor
permutations needed to place $U_x$ on $B_xG_x$ are understood throughout.
Since
$\Phi_{R\widehat S}\otimes I_{E_x}$ is a projection and unitary conjugation
preserves projections, each $P_x$ is a projection on
$\mathcal H_{\mathrm{tot}}$.

By Eq.~\eqref{eq:sm-marginal-invariance} the reduced state seen by every
placement is the same, so
\begin{equation}
  F_n=\trace\left[P_x\widetilde\omega\right]
  \qquad\text{for every }x\in\{0,1\}^n ,
  \label{eq:sm-common-fidelity}
\end{equation}
and therefore, with
\begin{equation}
  A:=2^{-n}\sum_{x\in\{0,1\}^n}P_x ,
  \label{eq:sm-average}
\end{equation}
we get $F_n=\trace[A\widetilde\omega]\le\lVert A\rVert_\infty$, using
$A\ge0$. Bounding $\lVert A\rVert_\infty$ is the content of
Secs.~\ref{sm:average} to \ref{sm:assembly}.

% =====================================================================
\section{From pairwise overlaps to the average projector}
\label{sm:average}
% =====================================================================

Let
$\mathbf 1\in\mathbb R^{2^n}$ be the all-ones vector and
$U:=2^{-n}\mathbf 1\mathbf 1^{\mathsf T}$, so that $U_{xy}=2^{-n}$ for all
$x,y$ and $U$ is the orthogonal projection onto the space spanned by $\mathbf 1$.

\begin{lemma}\label{lem:sm-average}
Let $\{P_x\}_{x\in\{0,1\}^n}$ be projections on a Hilbert space and
$A:=2^{-n}\sum_xP_x$. Suppose there is a real symmetric matrix
$W\in\mathbb R^{2^n\times2^n}$, indexed by $\{0,1\}^n$, with
\begin{equation}
  W\mathbf 1=\mathbf 1,
  \qquad
  \lVert W-U\rVert_\infty\le\zeta,
  \qquad
  \max_x\sum_y|W_{xy}|\le L,
  \label{eq:sm-W-hypotheses}
\end{equation}
and such that $W_{xy}\ne0$ implies
$\lVert P_xP_y\rVert_\infty\le\eta$. Then
\begin{equation}
  \lVert A\rVert_\infty\le\zeta+\sqrt{L\eta}.
  \label{eq:sm-average-bound}
\end{equation}
\end{lemma}

\begin{proof}
Since each $P_x$ is a projection, $0\le A\le I$, so
$\lambda:=\lVert A\rVert_\infty\in[0,1]$ is the largest eigenvalue of $A$;
let $\ket v$ be a corresponding unit eigenvector. Put
\begin{equation}
  \ket{y_x}:=2^{-n/2}P_x\ket v ,
  \qquad
  G_{xy}:=\braket{y_x|y_y},
\end{equation}
so that $G\ge0$ is a Gram matrix. Using $P_x^2=P_x$,
\begin{equation}
  \trace G=\sum_x\braket{y_x|y_x}
  =2^{-n}\sum_x\bra vP_x\ket v=\bra vA\ket v=\lambda ,
  \label{eq:sm-trG}
\end{equation}
and, since $\sum_xP_x=2^nA$,
\begin{equation}
  \sum_{x,y}G_{xy}=2^{-n}\bra v\Bigl(\sum_xP_x\Bigr)^2\ket v
  =2^{n}\bra vA^2\ket v=2^{n}\lambda^2 ,
  \qquad\text{hence}\qquad
  \trace[UG]=2^{-n}\sum_{x,y}G_{xy}=\lambda^2 .
  \label{eq:sm-trUG}
\end{equation}

Because $W$ is symmetric with $W\mathbf 1=\mathbf 1$ we have
$WU=UW=U$ and $U^2=U$, so
\begin{equation}
  (I-U)(W-U)(I-U)=W-U .
  \label{eq:sm-centering}
\end{equation}
Set $H:=(I-U)G(I-U)\ge0$. By Eqs.~\eqref{eq:sm-trG}
and~\eqref{eq:sm-trUG},
\begin{equation}
  \trace H=\trace\left[(I-U)G\right]=\lambda-\lambda^2\ge0 .
\end{equation}
Writing $S:=\sum_{x,y}W_{xy}G_{xy}$ and using the symmetry of $W$ to
identify $S=\trace[WG]$, Eq.~\eqref{eq:sm-centering} and cyclicity give
\begin{equation}
  S=\trace[UG]+\trace[(W-U)G]
   =\lambda^2+\trace[(W-U)H]
   \ \ge\ \lambda^2-\lVert W-U\rVert_\infty\trace H
   \ \ge\ \lambda^2-\zeta\left(\lambda-\lambda^2\right),
  \label{eq:sm-lower}
\end{equation}
where the first inequality is H\"older's inequality
$|\trace[XH]|\le\lVert X\rVert_\infty\lVert H\rVert_1$ together with
$\lVert H\rVert_1=\trace H$ for $H\ge0$.

On the other hand, $G_{xy}=2^{-n}\bra vP_xP_y\ket v$, so
\begin{equation}
  |S|\le\sum_{x,y}|W_{xy}|\,|G_{xy}|
   \le2^{-n}\sum_{x,y}|W_{xy}|\,\lVert P_xP_y\rVert_\infty
   \le2^{-n}\eta\sum_{x,y}|W_{xy}|
   \le\eta L ,
  \label{eq:sm-upper}
\end{equation}
where the third step uses that $W_{xy}\ne0$ only for pairs with
$\lVert P_xP_y\rVert_\infty\le\eta$, and the fourth uses
$\sum_{x,y}|W_{xy}|\le2^nL$.

Combining Eqs.~\eqref{eq:sm-lower} and~\eqref{eq:sm-upper} and discarding
the nonnegative term $\zeta\lambda^2$ gives
$\lambda^2-\zeta\lambda\le L\eta$. The positive root of the corresponding
quadratic yields
\begin{equation}
  \lambda\le\frac{\zeta+\sqrt{\zeta^2+4L\eta}}{2}\le\zeta+\sqrt{L\eta},
\end{equation}
using $\sqrt{a+b}\le\sqrt a+\sqrt b$.
\end{proof}

% =====================================================================
\section{Construction of the weight matrix}
\label{sm:weight}
% =====================================================================

The matrix $W$ is built from a low-degree approximation provided by Ref.~\cite{deWolfNote} to the Boolean
function $\operatorname{NOR}_n$, defined by $\operatorname{NOR}_n(0)=1$ and
$\operatorname{NOR}_n(x)=0$ for $x\ne0$.

\begin{theorem}[de Wolf~\cite{deWolfNote}]\label{thm:sm-approxdeg}
There is a universal constant $C_{\mathrm{app}}\ge1$ such that for every
$\varepsilon\in[2^{-n},1/3]$ there is a real multilinear polynomial $p$ on
$\{0,1\}^n$ of degree at most
$C_{\mathrm{app}}\bigl(\sqrt n+\sqrt{n\ln(1/\varepsilon)}\bigr)$
with $|p(x)-\operatorname{NOR}_n(x)|\le\varepsilon$ for all
$x\in\{0,1\}^n$.
\end{theorem}

\begin{lemma}\label{lem:sm-weight}
Let $C_{\mathrm{app}}$ be as in Theorem~\ref{thm:sm-approxdeg} and set
$a:=1/(4C_{\mathrm{app}}^2)$. For every integer $t$ with
\begin{equation}
  4C_{\mathrm{app}}\sqrt n\le t\le\frac n2
  \label{eq:sm-t-range}
\end{equation}
there is a real symmetric matrix $W\in\mathbb R^{2^n\times2^n}$, indexed by
$\{0,1\}^n$, satisfying
\begin{align}
  &W\mathbf 1=\mathbf 1, \label{eq:sm-W1}\\
  &W_{xy}\ne0\ \Longrightarrow\ s(x,y)\le t, \label{eq:sm-W2}\\
  &\lVert W-U\rVert_\infty\le2\exp\!\left(-\frac{at^2}{n}\right),
  \label{eq:sm-W3}\\
  &\max_x\sum_y|W_{xy}|\le\left(\sum_{j=0}^{t}\binom nj\right)^{1/2}.
  \label{eq:sm-W4}
\end{align}
\end{lemma}

\begin{proof}
Set $\varepsilon:=\exp\bigl(-t^2/(4C_{\mathrm{app}}^2n)\bigr)$. The upper bound in
Eq.~\eqref{eq:sm-t-range} gives
$t^2/(4C_{\mathrm{app}}^2n)\le n/(16C_{\mathrm{app}}^2)< n\ln2$, hence $\varepsilon\ge2^{-n}$; the lower
bound gives $t^2/(4C_{\mathrm{app}}^2n)\ge4$, hence $\varepsilon\le e^{-4}<1/3$. So
Theorem~\ref{thm:sm-approxdeg} applies and its degree bound reads
\begin{equation}
  C_{\mathrm{app}}\sqrt n+C_{\mathrm{app}}\sqrt{n\ln(1/\varepsilon)}
  =C_{\mathrm{app}}\sqrt n+C_{\mathrm{app}}\sqrt{n\cdot\frac{t^2}{4C_{\mathrm{app}}^2n}}
  =C_{\mathrm{app}}\sqrt n+\frac t2\le\frac t4+\frac t2<t .
\end{equation}
Let $p$ be the resulting polynomial. Since $p(0)\ge1-\varepsilon>0$ we may
normalize $q:=p/p(0)$, which is multilinear of degree at most $t$ and obeys
\begin{equation}
  q(0)=1,
  \qquad
  |q(z)|\le\frac{\varepsilon}{1-\varepsilon}\le2\varepsilon
  \quad(z\ne0).
  \label{eq:sm-qbounds}
\end{equation}

Write $\widehat q(e):=2^{-n}\sum_{x}q(x)(-1)^{e\cdot x}$ for the normalized
Hadamard transform, with inverse $q(x)=\sum_e\widehat q(e)(-1)^{e\cdot x}$.

\emph{Claim: $\widehat q(e)=0$ whenever $|e|>t$.} It suffices to check this
for a monomial $m_T(x)=\prod_{i\in T}x_i$ with $|T|\le t$. Its Hadamard transform yields
\begin{equation}
  \widehat{m_T}(e)
  =2^{-n}\prod_{i\in T}\Bigl(\sum_{x_i\in\{0,1\}}x_i(-1)^{e_ix_i}\Bigr)
  \prod_{i\notin T}\Bigl(\sum_{x_i\in\{0,1\}}(-1)^{e_ix_i}\Bigr).
\end{equation}
The second product vanishes unless $e_i=0$ for every $i\notin T$, i.e.
unless $\operatorname{supp}(e)\subseteq T$, in which case $|e|\le|T|\le t$.
This proves the claim.

Now define the matrix $W$ via (using  $\mathds 1=(1,\dots,1)$ and $\oplus$ addition modulo two)
\begin{equation}
  W_{xy}:=\widehat q(\mathds 1\oplus x\oplus y).
  \label{eq:sm-W-def}
\end{equation}
The matrix is real because $q$ is, and symmetric because
$x\oplus y=y\oplus x$. The string $\mathds1\oplus x\oplus y$ has a one
exactly at the coordinates where $x$ and $y$ agree, so
$|\mathds1\oplus x\oplus y|=s(x,y)$; the claim therefore gives
Eq.~\eqref{eq:sm-W2}. It also gives the crucial diagonal identity
\begin{equation}
  W_{xx}=\widehat q(\mathds1)=0,
  \label{eq:sm-W-diagonal}
\end{equation}
because $|\mathds1|=n>t$. For fixed $x$ the map
$y\mapsto\mathds1\oplus x\oplus y$
is a bijection of $\{0,1\}^n$, so
\begin{equation}
  \sum_yW_{xy}=\sum_e\widehat q(e)=q(0)=1,
\end{equation}
which is Eq.~\eqref{eq:sm-W1}.

For the spectral bound, let $\rchi_z\in\mathbb R^{2^n}$ have entries
$(\rchi_z)_x=(-1)^{x\cdot z}$; these $2^n$ vectors are mutually orthogonal.
Substituting $e=\mathds1\oplus x\oplus y$ and using
$(-1)^{(a\oplus b)\cdot z}=(-1)^{a\cdot z}(-1)^{b\cdot z}$,
\begin{equation}
  (W\rchi_z)_x=\sum_e\widehat q(e)(-1)^{(\mathds1\oplus x\oplus e)\cdot z}
  =(-1)^{|z|}(-1)^{x\cdot z}\sum_e\widehat q(e)(-1)^{e\cdot z}
  =(-1)^{|z|}q(z)\,(\rchi_z)_x .
\end{equation}
Also $(U\rchi_z)_x=2^{-n}\sum_y(-1)^{y\cdot z}=\delta_{z,0}$, so
$U\rchi_z=\delta_{z,0}\rchi_z$ because $\rchi_0=\mathbf 1$. Hence $W$ and $U$
are simultaneously diagonalized by $\{\rchi_z\}$, with $W-U$ having
eigenvalue $q(0)-1=0$ at $z=0$ and $(-1)^{|z|}q(z)$ otherwise. By
Eq.~\eqref{eq:sm-qbounds},
\begin{equation}
  \lVert W-U\rVert_\infty=\max_{z\ne0}|q(z)|\le2\varepsilon
  =2\exp\!\left(-\frac{t^2}{4C_{\mathrm{app}}^2n}\right),
\end{equation}
which is Eq.~\eqref{eq:sm-W3} with $a=1/(4C_{\mathrm{app}}^2)$.

Finally, using the same bijection, then the claim, then
Cauchy-Schwarz,
\begin{equation}
  \sum_y|W_{xy}|=\sum_e|\widehat q(e)|
  =\sum_{e:\,|e|\le t}|\widehat q(e)|
  \le\left(\sum_e|\widehat q(e)|^2\right)^{1/2}
  \left(\sum_{j=0}^{t}\binom nj\right)^{1/2}.
\end{equation}
Parseval's identity and Eq.~\eqref{eq:sm-qbounds} give
\begin{equation}
  \sum_e|\widehat q(e)|^2=2^{-n}\sum_zq(z)^2
  \le2^{-n}\left(1+2^n\cdot4\varepsilon^2\right)
  =2^{-n}+4\varepsilon^2\le\frac12+4e^{-8}<1
\end{equation}
for $n\ge1$, which yields Eq.~\eqref{eq:sm-W4}.
\end{proof}

% =====================================================================
\section{Overlap of two decoder placements}
\label{sm:overlap}
% =====================================================================

Finally, we have to bound the overlap of any two $P_x$ and $P_y$. The two placements interact only through the
outputs they share, and the shared dimension is what limits how much
entanglement both can recover.

\begin{lemma}\label{lem:sm-overlap}
For all distinct $x,y\in\{0,1\}^n$, the projections of
Eq.~\eqref{eq:sm-projector} satisfy
\begin{equation}
  \lVert P_xP_y\rVert_\infty
  \le\min\left\{1,\frac{d^{\,s(x,y)}}{M}\right\},
  \qquad d=|B| .
  \label{eq:sm-overlap}
\end{equation}
\end{lemma}

\begin{proof}
The bound by one is immediate. Fix $x\ne y$ and group the systems shared and
used exclusively by the two placements as
\begin{equation}
  C_{xy}:=\bigotimes_{j:\,x_j=y_j}B_{j,x_j},\quad
  X:=G_x\otimes\bigotimes_{j:\,x_j\ne y_j}B_{j,x_j},\quad
  Y:=\bigotimes_{j:\,x_j\ne y_j}B_{j,y_j}\otimes G_y .
\end{equation}
After canonical swap identifications, write
$B_xG_x\simeq X C_{xy}$ and $B_yG_y\simeq C_{xy} Y$, and use the fixed
global order $R\,X\, C_{xy}\,Y$. The shared dimension is
$D:=|C_{xy}|=d^{s(x,y)}$; all other systems carry identities and may be omitted.

For two projections $P_x,P_y$,
\begin{equation}
  \lVert P_xP_y\rVert_\infty=
  \sup_{\substack{P_x\ket{\psi}=\ket{\psi},\ P_y\ket{\phi}=\ket{\phi}\\
                   \braket{\psi|\psi} = \braket{\phi|\phi} = 1}}
  |\braket{\psi|\phi}|.
  \label{eq:sm-projection-angle}
\end{equation}
Choose an orthonormal basis $\{\ket c\}_{c=1}^{D}$ of $C_{xy}$ and define
\begin{align}
  A_{i,c}&:=
  (\bra i_{\widehat S}\otimes I_{E_x})
  U_x(I_X\otimes\ket c_{C_{xy}}):X\to E_x,\nonumber\\
  B_{i,c}&:=
  (\bra i_{\widehat S}\otimes I_{E_y})
  U_y(\ket c_{C_{xy}}\otimes I_Y):Y\to E_y .
  \label{eq:sm-blocks}
\end{align}
Taking the $(c,c)$ block of $U_x^\dagger U_x=I_{XC}$ and of
$U_y^\dagger U_y=I_{CY}$ gives
\begin{equation}
  \sum_{i=1}^{M}A_{i,c}^\dagger A_{i,c}=I_X,
  \qquad
  \sum_{i=1}^{M}B_{i,c}^\dagger B_{i,c}=I_Y .
  \label{eq:sm-block-unitarity}
\end{equation}
Define $V_x$ and $V_y$ by the two right-hand sides below. Because
$U_x,U_y$ are unitary and $\ket{\Phi_M}$ is normalized, these maps are
isometries with the same  ranges as  $P_x$ and
$P_y$, respectively. Hence every pair of normalized states in
the two ranges can be written, for some normalized
$\ket\alpha$ on $ E_x Y$ and
$\ket\beta$ on $ X E_y$, as
\begin{align}
  \ket\psi=V_x\ket\alpha&=\frac1{\sqrt M}\sum_{i,c}\ket i_R\otimes
  (A_{i,c}^\dagger\otimes\ket c_{C_{xy}}\otimes I_Y)\ket\alpha,\nonumber\\
  \ket\phi=V_y\ket\beta&=\frac1{\sqrt M}\sum_{i,c}\ket i_R\otimes
  (I_X\otimes\ket c_{C_{xy}}\otimes B_{i,c}^\dagger)\ket\beta .
  \label{eq:sm-range-expansions}
\end{align}
Contracting the $R$ and $C_{xy}$ indices yields
\begin{equation}
  \braket{\psi|\phi}=\frac1M\sum_{c=1}^{D}\bra\alpha K_c\ket\beta,
  \qquad
  K_c:=\sum_{i=1}^{M}A_{i,c}\otimes B_{i,c}^\dagger:
  XE_y\to E_x Y .
  \label{eq:sm-overlap-expansion}
\end{equation}
To bound $K_c$, introduce an $M$-dimensional register $L$ and set
\begin{equation}
  R_c:=\sum_i\ket i_L\otimes A_{i,c}\otimes I_{E_y},
  \qquad
  S_c:=\sum_i\bra i_L\otimes I_{E_x}\otimes B_{i,c}^\dagger .
\end{equation}
Then $S_cR_c=K_c$, while Eq.~\eqref{eq:sm-block-unitarity} gives
$R_c^\dagger R_c=I_{XE_y}$ and $S_cS_c^\dagger=I_{E_xY}$. Hence $R_c$ is
an isometry, $S_c$ is a co-isometry, and $\lVert K_c\rVert_\infty\le1$.
Equation~\eqref{eq:sm-overlap-expansion} gives
$|\braket{\psi|\phi}|\le D/M$. Taking the supremum in
Eq.~\eqref{eq:sm-projection-angle} and combining with the trivial bound by
one proves Eq.~\eqref{eq:sm-overlap}.
\end{proof}

% =====================================================================
\section{Proof of Proposition~\ref{prop:finite-block} and Theorem~\ref{thm:antidegradable}}
\label{sm:assembly}
% =====================================================================
Finally, we can put all ingredients together and prove Proposition~\ref{prop:finite-block} in the main text.
\begin{proof}[Proof of Proposition~\ref{prop:finite-block}]
Let $t$ satisfy Eq.~\eqref{eq:sm-t-range} and let $W$ be the matrix of
Lemma~\ref{lem:sm-weight}. Equations~\eqref{eq:sm-W2} and
\eqref{eq:sm-W-diagonal} show that $W_{xy}\ne0$ forces both $x\ne y$ and
$s(x,y)\le t$, so Lemma~\ref{lem:sm-overlap} bounds
$\lVert P_xP_y\rVert_\infty\le\eta:=\min\{1,d^{\,t}/M\}$ on the support of
$W$. Lemma~\ref{lem:sm-average} with
$\zeta=2\exp(-at^2/n)$ and $L=\bigl(\sum_{j\le t}\binom nj\bigr)^{1/2}$
then gives
\begin{equation}
  F_n\le\lVert A\rVert_\infty\le
  2\exp\!\left(-\frac{at^2}{n}\right)
  +\left(\sum_{j=0}^{t}\binom nj\right)^{1/4}
  \min\left\{1,\frac{d^{\,t}}{M}\right\}^{1/2},
\end{equation}
which is Eq.~\eqref{eq:finite-block}. The constants $a=1/(4C_{\mathrm{app}}^2)$ and $C_{\mathrm{app}}$ are
universal, and no property of $\rho_{RA'^{\,n}}$ or $\mathcal D_n$ was used
beyond Eq.~\eqref{eq:sm-marginal-invariance}.
\end{proof}

Theorem~\ref{thm:antidegradable} follows by the choice
$t=\lfloor\delta n\rfloor$ made in the main text and via the calculation starting with Eq.~\eqref{eq:delta-choice}.

% =====================================================================
\section{Proof of Lemma~\ref{lem:symmetric-rate}}
\label{sm:symmetric-rate}
% =====================================================================

\begin{proof}
Let $\mathcal M$ be antidegradable with output dimension $d_{\mathcal M}$, so
Proposition~\ref{prop:finite-block} applies to it. Set
\begin{equation}
  K:=\sqrt{\frac{2\ln2}{a}},
  \qquad
  c_0:=\frac{1}{16K^2},
\end{equation}
and let $0<c\le c_0$. Consider any generalized $(n,M)$ Bell-target scheme
for $\mathcal M$ with error at most $1-\lambda_n$, where
$\lambda_n\ge2^{-cn}/10$. Then
\begin{equation}
  F_n=1-\varepsilon_n^2\ge1-(1-\lambda_n)^2
  =\lambda_n(2-\lambda_n)\ge\lambda_n\ge\frac{2^{-cn}}{10}.
  \label{eq:sm-fidelity-lower}
\end{equation}

Put $\theta:=K\sqrt c\le1/4$ and $t:=\lceil\theta n\rceil$. Then
$t/n\to\theta$, and for all sufficiently large $n$ one has
$4C_{\mathrm{app}}\sqrt n\le t\le n/2$, so Eq.~\eqref{eq:t-range} applies. The choice of
$K$ is
made precisely so that
\begin{equation}
  a\theta^2\log e=aK^2c\log e=2c ,
  \qquad\text{whence}\qquad
  2\exp\!\left(-\frac{at^2}{n}\right)\le2\cdot2^{-2cn}
  \le\frac12\cdot\frac{2^{-cn}}{10}
\end{equation}
once $2^{cn}\ge40$. For such $n$ the first term of
Eq.~\eqref{eq:finite-block} accounts for at most half of the lower bound
in Eq.~\eqref{eq:sm-fidelity-lower}, so the second term must satisfy
\begin{equation}
  \left(\sum_{j=0}^{t}\binom nj\right)^{1/4}
  \min\left\{1,\frac{d_{\mathcal M}^{\,t}}{M}\right\}^{1/2}
  \ \ge\ \frac{2^{-cn}}{20}.
\end{equation}
If the minimum equals one then $M\le d_{\mathcal M}^{\,t}$ and
$\frac1n\log M\le(t/n)\log d_{\mathcal M}$, whose limsup is below
$g_{\mathcal M}(c)$. Otherwise, using
$\sum_{j\le t}\binom nj\le2^{nh_2(t/n)}$ and rearranging,
\begin{equation}
  M\le400\,d_{\mathcal M}^{\,t}\,2^{2cn}\,2^{nh_2(t/n)/2},
  \qquad\text{so}\qquad
  \frac1n\log M\le\frac tn\log d_{\mathcal M}
  +\frac12h_2(t/n)+2c+\frac1n\log400.
\end{equation}
Taking $n\to\infty$ and using $t/n\to\theta=K\sqrt c$ gives
Eq.~\eqref{eq:symmetric-rate-bound}. The term $2c$ arises because the square root in the second term of
Eq.~\eqref{eq:finite-block} is squared when solving for $M$; fixed constants,
including $400$, contribute only $o(1)$. Finally
$g_{\mathcal M}(c)\to0$ as $c\downarrow0$ because
$h_2(u)\to0$ as $u\downarrow0$.
\end{proof}

% =====================================================================
\section{Flagged channels}
\label{sm:flagged}
% =====================================================================
Here, we prove Proposition~\ref{prop:flagged} in the main text.

We first establish an elementary dimension bound used for the degenerate case
$q=1$.

\begin{lemma}\label{lem:sm-dimension}
Let $\Lambda:\bar B\to\widehat S$ be a channel with $|\bar B|=D$ and
$|\widehat S|=|R|=M$, and let $\sigma_{R\bar B}$ be any state. Then
\begin{equation}
  \bra{\Phi_M}\left(\operatorname{id}_R\otimes\Lambda\right)(\sigma)
  \ket{\Phi_M}\le\frac DM .
\end{equation}
\end{lemma}

\begin{proof}
Let $\Lambda^{*}$ be the adjoint of $\Lambda$ with respect to the
Hilbert-Schmidt inner product; it is completely positive and unital. Then
\begin{equation}
  \bra{\Phi_M}(\operatorname{id}\otimes\Lambda)(\sigma)\ket{\Phi_M}
  =\trace\left[T\,\sigma\right],
  \qquad
  T:=(\operatorname{id}_R\otimes\Lambda^{*})(\Phi_{R\widehat S})\ge0 .
\end{equation}
Since partial trace over $R$ commutes with a map acting on the other
factor, $\trace_R T=\Lambda^{*}(\trace_R\Phi_{R\widehat S})
=\Lambda^{*}(I_{\widehat S}/M)=I_{\bar B}/M$, so
$\trace T=D/M$. As $T\ge0$ we get
$\trace[T\sigma]\le\lVert T\rVert_\infty\le\trace T=D/M$.
\end{proof}

\begin{proof}[Proof of Proposition~\ref{prop:flagged}]
If $w=0$ the claim is Lemma~\ref{lem:sm-dimension} applied to
$\mathcal K_0^{\otimes n}$, and if $q=0$ the channel is
$\mathcal K_1$ tensored with a constant flag, so
Theorem~\ref{thm:antidegradable} applies and the threshold $q\log d_0=0$ is
correct. Assume therefore $0<q<1$, and write $d_1:=|B^{(1)}|$.

The flag register $C$ is classical and is part of the channel output, so the
$n$-fold output state is
$\sum_{z\in\{0,1\}^n}p_z\,\sigma^{(z)}_{R\mathbf B}\otimes\ketbra zz_C$ with
$p_z=q^{k(z)}w^{m(z)}$, where $k(z):=|\{j:z_j=0\}|$, $m(z):=n-k(z)$, and
$\sigma^{(z)}$ is the output of $\bigotimes_j\mathcal K_{z_j}$ on
$\rho_{RA'^{\,n}}$. Since the fidelity is linear in the state,
\begin{equation}
  F_n=\sum_{z\in\{0,1\}^n}p_z f_z ,
  \label{eq:sm-flag-average}
\end{equation}
where $f_z$ is the fidelity achieved by the decoder conditioned on the flag
string $z$. Formally, if $\mathcal D$ is an arbitrary decoder on the
block-diagonal flagged output, then
$\mathcal D_z(\tau):=\mathcal D(\tau\otimes\ketbra zz)$ is a valid channel
for each fixed $z$, and $f_z$ is its Bell overlap.

Fix $z$ and abbreviate $k=k(z)$, $m=m(z)$. Apply the swap-symmetric
two-output extension of $\mathcal K_1$ at the $m$ positions with $z_j=1$,
leaving the $k$ branch-$0$ outputs untouched. Every $x\in\{0,1\}^m$ selects
one output at each extended position, and by the same marginal-invariance
argument as in Eq.~\eqref{eq:sm-marginal-invariance} each of the $2^m$
resulting placements attains the same conditional fidelity $f_z$. Repeating
the construction of Sec.~\ref{sm:setup}  gives projections
$P^{(z)}_x$ with $f_z=\trace[P^{(z)}_x\widetilde\omega^{(z)}]$ for every
$x$.

Two placements $x,y$ now share the $k$ branch-$0$ outputs \emph{in addition}
to the $s(x,y)$ extended positions where they agree, so the shared subsystem in
Lemma~\ref{lem:sm-overlap} has dimension
$d_0^{\,k}d_1^{\,s(x,y)}$ and that lemma gives
\[
\|P_x^{(z)}P_y^{(z)}\|_\infty
 \leq
 \min\left\{1,
 \frac{d_0^{\,k}d_1^{\,s(x,y)}}{M_n}\right\}.
\]. Lemmas~\ref{lem:sm-average}
and~\ref{lem:sm-weight} then yield, for every integer $t$
with $4C_{\mathrm{app}}\sqrt m\le t\le m/2$,
\begin{equation}
  f_z\le2\exp\!\left(-\frac{at^2}{m}\right)
  +\left(\sum_{j=0}^{t}\binom mj\right)^{1/4}
  \min\left\{1,\frac{d_0^{\,k}d_1^{\,t}}{M_n}\right\}^{1/2}.
  \label{eq:sm-conditional-finite-block}
\end{equation}

Now fix $r>q\log d_0$ and choose $\epsilon,\delta>0$ with
\begin{equation}
  (q+\epsilon)\log d_0+\delta\log d_1+\tfrac12h_2(\delta)<r,
  \qquad
  \epsilon<1-q,
  \qquad
  \delta<\frac{1-q-\epsilon}{2},
  \label{eq:sm-eps-delta}
\end{equation}
which is possible because the left-hand side of the first inequality tends
to $q\log d_0<r$ as $\epsilon,\delta\downarrow0$. The random variable
$k(z)$ is binomial with mean $qn$, so by Hoeffding's inequality
\begin{equation}
  \Pr\left[k(z)>(q+\epsilon)n\right]\le e^{-2\epsilon^2n}.
  \label{eq:sm-chernoff}
\end{equation}
On the complementary event $m\ge(1-q-\epsilon)n$. Set
$t:=\lfloor\delta n\rfloor$; then $t\le\delta n<m/2$ by
Eq.~\eqref{eq:sm-eps-delta}; moreover $t\ge\delta n/2$ and, since $n\geq m$,
$t\ge4C_{\mathrm{app}}\sqrt m$ for all large $n$, so
Eq.~\eqref{eq:sm-conditional-finite-block} applies. Using $m\le n$,
$\sum_{j\le t}\binom mj\le\sum_{j\le t}\binom nj\le2^{nh_2(\delta)}$, and
$M_n\ge2^{nr}$, we obtain uniformly over all typical $z$
\begin{equation}
  f_z\le2\cdot2^{-a\delta^2n\log e/4}
  +2^{-\frac n2\left[r-(q+\epsilon)\log d_0-\delta\log d_1
  -\frac12h_2(\delta)\right]},
\end{equation}
and both exponents are strictly positive by Eq.~\eqref{eq:sm-eps-delta}.
Bounding $f_z\le1$ on the atypical event and inserting
Eq.~\eqref{eq:sm-chernoff} into Eq.~\eqref{eq:sm-flag-average} gives
$F_n\le e^{-2\epsilon^2n}+2^{-\gamma'n}$ for some $\gamma'>0$. Hence every
$r>q\log d_0$ is an exponential strong-converse rate for $\mathcal F$.
\end{proof}

\end{document}